\documentclass{article}

\usepackage{arxiv}
\usepackage{amsmath, amsthm}

\usepackage[utf8]{inputenc} 
\usepackage[T1]{fontenc}    
\usepackage{hyperref}       
\usepackage{url}            
\usepackage{booktabs}       
\usepackage{amsfonts}       
\usepackage{microtype}      
\usepackage{natbib}
\usepackage{shortcuts}

\usepackage{multibib} 

\usepackage{tikz}
\usepackage{tkz-euclide}
\usetkzobj{all}
\usepackage{float}

\usepackage{amsfonts, amssymb}

\title{Multivariate one-sided testing in matched observational studies as an adversarial game}

\author{
	Peter L. Cohen \\
	Operations Research Center\\
	Massachusetts Institute of Technology\\
	1 Amherst Street\\
	Cambridge Massachusetts 02142, U.S.A\\
	\texttt{plcohen@mit.edu} \\
	\And
	Matt A. Olson \\
	The Voleon Group\\
	2170 Dwight Way\\
	Berkeley, CA 94704, U.S.A \\
	\texttt{molson@voleon.com} \\
	\And
	Colin B. Fogarty \\
	Operations Research and Statistics Group\\
	Massachusetts Institute of Technology\\
	100 Main Street\\
	Cambridge, Massachusetts 02142,  U.S.A\\
	\texttt{cfogarty@mit.edu} 
}

\begin{document}
	
	\maketitle
	
	\begin{abstract}
		We present a multivariate one-sided sensitivity analysis for matched observational studies, appropriate when the researcher has specified that a given causal mechanism should manifest itself in effects on multiple outcome variables in a known direction. The test statistic can be thought of as the solution to an adversarial game, where the researcher determines the best linear combination of test statistics to combat nature's presentation of the worst-case pattern of hidden bias. The corresponding optimization problem is convex, and can be solved efficiently even for reasonably sized observational studies. Asymptotically the test statistic converges to a chi-bar-squared distribution under the null, a common distribution in order restricted statistical inference. The test attains the largest possible design sensitivity over a class of coherent test statistics, and facilitates one-sided sensitivity analyses for individual outcome variables while maintaining familywise error control through is incorporation into closed testing procedures. 
		
	\end{abstract}
	
	\section{On Multiplicity and Causality}
	Controlled randomization protects empirical evidence against a host of counterclaims. A significant finding may well be due to random chance alone, but cannot be dismissed on the grounds of biases unaccounted for by the study's design. Observational evidence provides no such assurance, and causal inference in observational studies involves ambiguity which randomization eschews: Is the association an effect, or is it bias from self-selection? Anticipating skepticism, a practitioner may take measures when planning an observational study to properly frame the debate, rendering certain criticism unwarranted should the practitioner's hypothesis be true. While ambiguity cannot be eliminated, quasi-experimental devices may be employed to help clarify the step from association to causation in observational studies; see \citet{sha02} and \citet{ros15quasi} for an overview. One such device, known alternatively as pattern specificity, multiple operationalism, or coherence, advocates that observational studies be designed with the objective of confirming a complex pattern of predictions made by the causal theory in question. This is in keeping with Fisher's notion of elaborate theories, which advocates that the practitioner ``envisage as many different consequences of [a causal hypothesis's] truth as possible, and plan observational studies to discover whether each of these consequences is found to hold" \citep[][\S 5, p. 252]{coc65}. Complex predictions imperil the practitioner's hypothesis, as doubt is cast should any prediction fail in the observational study at hand. Should the evidence prove coherent with the theory's predictions, fortification is provided as attributing a complex pattern to hidden bias requires that hidden bias could reproduce the particular pattern of association. 
	
	One way in which a theory can be made elaborate is through predicting that an intervention will affect multiple outcome variables in a prespecified direction. While the practitioner hopes that each prediction holds, should certain predictions fail she would regardless like to quantify which components came to fruition as a means of refining understanding of the mechanism in question. With this comes the attending issues of multiple comparisons. Concerns over a loss in power from multiplicity control may lead practitioners to instead investigate the outcome they believe \textit{a priori} will be most affected, reducing the extent to which Fisher's advice is followed. 
	
	The qualitative benefits of multiple outcomes in observational studies are thus at odds with the statistical corrections they require. This tension exists not only when assuming no hidden bias, but also in the sensitivity analysis where the researcher quantifies the magnitude of hidden bias required to overturn the study's conclusions. In what follows, we present a new method for sensitivity analysis in multivariate one-sided testing, appropriate when the researcher anticipates a particular direction of effects for multiple outcome variables. The test adaptively combines outcome-specific test statistics, has the optimal design sensitivity over a class of multivariate tests respecting coherence, and leads to substantial improvements in power when the researcher's prediction proves correct. The method greatly attenuates the impact of multiplicity control on power for testing individual outcome variables through its use in closed testing procedures \citep{mar76}, facilitating the analysis of multiple outcomes for demonstrating coherence.

	\section{Hidden Bias in Matched Observational Studies}\label{sec: Sensitivity Analysis}
	\subsection{A finely stratified experiment with multiple outcomes} \label{sec:Notation}
	There are $I$ independent strata, the $i$th of which contains $n_i\geq 2$ individuals. Individual $j$ in stratum $i$ has a $P$-dimensional vector of observed covariates $x_{ij}$, along with an unobserved covariate $u_{ij}$, $0\leq u_{ij}\leq 1$. The strata are formed such that $x_{ij}\approx x_{ij'}$ for any two individuals $j\neq j'$ in stratum $i$.  We take $Z_{ij}$ as the indicator of treatment for the $j$th individual in stratum $i$, such that $Z_{ij}=1$ if assigned to treatment and $Z_{ij}=0$ otherwise. Each strata contains one treated individual and $n_i-1$ controls such that $\sum_{j=1}^{n_i}Z_{ij} = 1$ $(i=1,...,I)$. See \citet{fog18mit} for more on this particular class of stratified experiments, referred to as finely stratified experiments. Forthcoming developments readily extend to full-matched observational studies; see \citet[Ex. 4.12]{R02} for details.
	
	Each individual has two vectors of potential outcomes of length $K$: the responses for each outcome variable under control $r_{Cij} = (r_{Cij1},...,r_{CijK})^T$, and the responses under treatment $r_{Tij} = (r_{Tij1},...,r_{TijK})^T$. The $K$-dimensional vector of treatment effects $\tau_{ij} = r_{Tij} - r_{Cij}$ is not observed; instead, we observe the vector $\vect{R}_{ij} = Z_{ij}\vect{r}_{Tij} + (1 - Z_{ij})\vect{r}_{Cij}$.  Let $Z = (Z_{11},...,Z_{In_{I}})^T$ be the lexicographically ordered vector of treatment assignments of length $N$, and let the analogous hold for $u$ along with $r_{Ck}, r_{Tk}$ and $R_k$ for $k=1,...,K$.  The $N\times K$ matrix with lexicographically ordered rows containing $R_{ij}^T$ is $R$. 
	
	Let $\F = \left\{\vect{r}_{Cij}, \vect{r}_{Tij}, \vect{x}_{ij}, \vect{u}_{ij}: i=1,...,I; j=1,...,n_i \right\}$ be a set containing the potential outcomes along with the measured and unmeasured covariates for each individual in the observational study. In what follows we consider inference conditional upon $\cF$, such that a generative model for the potential outcomes is neither assumed nor required. Let $\Omega = \{z: \sum_{j=1}^{n_i} z_{ij}=1; i=1,...,I\}$ be the set of $\prod_{i=1}^{I}n_i$ treatment assignments adhering to the stratified design, and let $\cZ = \{Z\in \Omega\}$ be the event that the observed treatment assignment satisfies this design. In a finely stratified experiment $\P(Z_{ij}=1\mid \cF, \cZ) = 1/n_i$ and $\P(Z=z\mid \cF, \cZ) = 1/|\Omega|$ where $|A|$ is the cardinality of the set $A$. 
	
	\subsection{A model for biased treatment assignment}
	
	Matched observational studies aim to mimic the finely stratified experiment described in \S \ref{sec:Notation}. Matching algorithms assign individuals to matched sets on the basis of observed covariates such that $x_{ij}\approx x_{ij'}$ for individuals $j$ and $j'$ in the same matched set $i$; see \citet{Han04} and \citet{zub12} among many for more on matching algorithms and the optimization problems underpinning them. A simple model for treatment assignment in an observational study states that before matching, individuals are assigned to treatment independently with unknown probabilities $\pi_{ij} = \P(Z_{ij}=1\mid \cF)$. While one may hope that $\pi_{ij}\approx \pi_{ij'}$ after matching, proceeding as such would be specious due to both the potential presence of unmeasured confounding \textcolor{black}{and residual imbalances on the observed covariates in each matched set}. The model of \citet[][Chapter 4]{R02} stipulates that individuals in the same matched set may differ in their odds of assignment to treatment by at most a factor of $\Gamma$, 
	\begin{equation}\label{eq:sensmodel}
	\frac{1}{\Gamma} \leq \frac{\pi_{ij}(1 - \pi_{ij'})}{\pi_{ij'}(1 - \pi_{ij})} \leq \Gamma.
	\end{equation} 
	
	The parameter $\Gamma$ controls the degree to which matching solely on observed covariates may have failed to align the assignment probabilities in each matched set. The value $\Gamma = 1$ returns a randomized finely stratified experiment, while $\Gamma>1$ allows for a tilt in the randomization distribution to a degree controlled by $\Gamma$. For instance, $\Gamma=2$ stipulates that individuals in the same matched set might truly differ in their odds of receiving the treatment by a factor of at most two. Returning attention to the matched structure by conditioning on $\cZ$, this model is equivalent to assuming \begin{equation}\label{eqn: treatment assignment distribution}
	\Prob{\vect{Z} = \vect{z} \given \F,\, \cZ} = \frac{\exp\left(\gamma \vect{z}^{T} \vect{u} \right)}{\sum_{b \in \Omega}\exp\left(\gamma \vect{b}^{T} \vect{u} \right)},
	\end{equation}
	where $\gamma = \log(\Gamma)$ and $u$ lies in the $N$-dimensional unit cube, call it $\mathcal{U}$, embodying both differences in unobserved covariates and latent discrepancies in observed covariates after matching; see \citet{ros95} or \citet[][Chapter 4]{R02} for a proof of this equivalence. 
	
	\subsection{Sensitivity analysis for a particular outcome}    
	Assume without loss of generality that the outcomes have been recorded such that positive values for the treatment effects $\tau_{ijk}$ are predicted by the causal theory under study. For each outcome variable, we consider tests of the null hypothesis of non-positive treatment effects,
	\begin{align*}
	H_k: r_{Tijk} \leq r_{Cijk}\;\;\; (i=1,...,I; j=1,...,n_i).
	\end{align*}
	$H_k$ is a composite null hypothesis. Elements of $H_k$ include the null of a non-positive constant effect for all individuals, $r_{Tijk} = r_{Cijk} +\delta_k$ for any scalar $\delta_k \leq 0$; and certain models of tobit effects, such as $r_{Cijk} = \max\{r_{Tijk}, 0\}$. Fisher's sharp null of no effect is $\delta_k=0$, thus representing the boundary of $H_k$.  \textcolor{black}{		The composite null $H_{k}$ is distinct from Neyman's weak null of no average treatment effect for the $k$th outcome variable. That said, both nulls allow for inference without prespecifying the particular pattern of effect heterogeneity. Neyman's null has been seen as a flexible way to test for existence of treatment effect while accommodating arbitrary effect heterogeneity. Unfortunately, testing Neyman's null on the $k$th outcome greatly constrains the test statistics available to the practitioner, requiring the use of a studentized difference-in-means or a regression-adjusted estimator \citep{RandTestsWeakNulls}. These test statistics have poor theoretical properties when used in a sensitivity analysis. The null $H_{k}$ is also more general than a sharp null, but can still be tested through randomization inference using statistics such as $m$-tests with better theoretical properties in the potential presence of hidden bias  \citep{ros07}.} \textcolor{black}{The null $H_k$ is not limited to continuous outcome variables, and can also be employed with ordinal outcomes. In fact, our method may be used with potential outcomes of any partially ordered set. See \S\ref{sec: Discussion} for further details.}
	
	We consider test statistics for each outcome variable which are effect increasing sum statistics. Sum statistics are statistics of the form $T_k(Z, R_k) = Z^Tq_k$ where $q_k = q_k(R_k)$ is a pre-specified function of the observed responses $R_k$. A test statistic is effect increasing if $T_k(z, r^*_k) \geq T_k(z, r_k)$ whenever $(2z_{ij}-1)(r^*_{ijk}-r_{ijk})\geq 0$ for all $i$ and $j$\textcolor{black}{, where $r^*_{ijk}$ denotes a different value of the potential outcome}. In words, this means that if every treated unit did better with $r^*_k$ than with $r_k$, and if every control did worse with $r^*_k$ than with $r_k$, then the test statistic corresponding to the observed outcomes $r^*_k$ would be larger than it would have been under $r_k$. Most familiar test statistics, including differences-in-means, rank tests, and $m$-tests are endowed with these properties; see \citet[][Chapter 2.4.4]{R02} and \citet[][\S 3.1]{ros16} for additional examples.

	If Fisher's sharp null is true then $R_{k} = r_{Ck}$, and hence $T_k(Z,R_k) = T_k(Z, r_{Ck})$. For a particular $\Gamma>1$, the test statistic's null distribution under Fisher's sharp null is   
	\begin{align}\label{eq:null}
	\P\{T_k(Z, r_{Ck})\geq v\mid \cF, \cZ\} &=  \sum_{z\in \Omega}1\{T_k(Z, r_{Ck})\geq v\}\frac{\exp\left(\gamma \vect{z}^{T} \vect{u} \right)}{\sum_{b \in \Omega}\exp\left(\gamma \vect{b}^{T} \vect{u} \right)},
	\end{align}   
	where $1(A)$ is an indicator that the condition $A$ was met. At $\Gamma=1$ (\ref{eq:null}) is simply the proportion of treatment assignments where the test statistic is greater than or equal to $v$, returning the usual randomization inference in a finely stratified experiment. For $\Gamma>1$ (\ref{eq:null}) is unknown due to its dependence on the nuisance vector $u$. A sensitivity analysis proceeds for a particular $\Gamma$ by maximizing (\ref{eq:null}) with $v=t_k$, the observed value of the test statistic for a particular $\Gamma$, resulting in the largest possible $p$-value for the desired inference subject to (\ref{eq:sensmodel}) holding at $\Gamma$. The practitioner then increases $\Gamma$ until the test no longer rejects the null hypothesis. This changepoint value of $\Gamma$ serves as a measure of how robust the study's finding was to unmeasured confounding. See \citet{gas00} and \citet{ros18} for large-sample approaches for conducting a sensitivity analysis for Fisher's sharp null with a single outcome variable under (\ref{eq:sensmodel}). Since $T_k$ is assumed effect increasing, the worst-case $p$-value for a sensitivity analysis for Fisher's sharp null attains the largest $p$-value over the composite null $H_k$. That is, a sensitivity analysis for Fisher's sharp null also provides a valid sensitivity analysis for $H_k$ \citep[Prop. 1]{CDM17}.
	

	
	\section{Sensitivity Analysis with multiple outcomes}
	\subsection{A directional global null hypothesis}
	There are $K$ hypotheses $H_1,...,H_K$, one each for the null of non-positive treatment effects for the $k$th outcome variable. We concern ourselves with a level-$\alpha$ sensitivity analysis for the global null hypothesis that all $K$ of these hypotheses are true,
	\begin{align}\label{eq:global}
	H_0 : \bigwedge_{k=1}^K H_k.
	\end{align} 
	
	Through closed testing \citep{mar76}, a valid sensitivity analysis for (\ref{eq:global}) also facilitates tests of the outcome-specific hypotheses $H_k$ while controlling the familywise error rate. See \citet[][\S 5]{FS16} for more on closed testing procedures applied to sensitivity analyses. 
	
	
	\subsection{Linear combinations of test statistics and their distribution}
	
	In what follows it is useful to define $\varrho_{ij} = \Prob{Z_{ij} = 1 \given \F, \cZ}$. Under the global null (\ref{eq:global}) and recalling that our test statistics are of the form $T_k = Z^Tq_k$ with $q_k$ fixed under the global null, the expectation $\mu(\varrho)$ and covariance $\Sigma(\varrho)$ for the vector of test statistics $T = (T_{1}, \ldots, T_{K})^{T}$ are
	\begin{equation*}
	\mu(\vect{\varrho})_{k} = \sum_{i = 1}^{I}\sum_{j = 1}^{n_{i}}q_{ijk}\varrho_{ij},\;\;
	\Sigma(\vect{\varrho})_{k,\ell} =  \sum_{i = 1}^{I} \left\{\sum_{j = 1}^{n_{i}}q_{ijk}q_{ij\ell}\varrho_{ij} - \left(\sum_{j = 1}^{n_{i}}q_{ijk}\varrho_{ij}\right)\left(\sum_{j = 1}^{n_{i}}q_{ij\ell}\varrho_{ij}\right)   \right\}. 
	\end{equation*}For a given vector of probabilities $\varrho$, under suitable conditions on the constants $q_{ijk}$ the distribution of $T$ is asymptotically multivariate normal through an application of the Cram\'er-Wold device. That is, for any fixed nonzero vector $\lambda = (\lambda_1,...,\lambda_K)^T$ the standardized deviate $\lambda^T\{T-\mu(\varrho)\}/\left\{\lambda^T\Sigma(\varrho)\lambda\right\}^{1/2}$ is asymptotically standard normal.
	
	The actual values of $\varrho$ are unknown to the practitioner due to their dependence on hidden bias. Instead, the constraints imposed by the sensitivity model (\ref{eq:sensmodel}) on $\varrho$ can be represented by a polyhedral set. For a particular $\Gamma$ this set, call it $\cP_\Gamma$, contains vectors $\varrho$ such that (i) $\varrho_{ij}\geq 0$ $(i=1,...,I; j=1,...,n_i)$; (ii) $\sum_{i=1}^{n_i} \varrho_{ij} = 1$ $(i=1,...,I)$; and (iii) $s_i \leq \varrho_{ij} \leq \Gamma s_i$ for some $s_i$ $(i=1,...,I; j=1,...,n_i)$. Conditions (i) and (ii) simply reflect that $\varrho_{ij}$ are probabilities, while the $s_i$ terms in (iii) arise from applying a Charnes-Cooper transformation \citep{CC62} to (\ref{eqn: treatment assignment distribution}). 
	
	\subsection{Multivariate sensitivity analysis through a two-person game}\label{sec:twoplayer}
	Let $t = (t_1,...,t_K)^T$ be the observed vector of test statistics. In this subsection only, suppose interest lies not in a test of (\ref{eq:global}), but rather in the narrower intersection null that Fisher's sharp null holds for all $K$ outcome variables. For fixed $\lambda$, a large-sample sensitivity analysis for Fisher's sharp null could be achieved by minimizing the standardized deviate $\lambda^T\{t-\mu(\varrho)\}/\{\lambda^T\Sigma(\varrho)\lambda\}^{1/2}$ over all $\varrho$ such that $\varrho \in \cP_\Gamma$, and assessing whether the minimal objective value exceeds the appropriate critical value from a standard normal.
	
	With $\lambda$ pre-specified, the sensitivity analysis imagines what would happen if the worst-case, adversarial bias at a given level of $\Gamma$ were present. If the practitioner fixes the linear combination $\lambda$ ahead of time, she has no further recourse against such adversarial attacks. The practitioner may instead consider a two-person game of the form
	\begin{align}\label{eq:twoplayer}
	{a}^*_{\Gamma, \Lambda} = \underset{\varrho\in\cP_\Gamma}{\min}\;\; \underset{\lambda \in \Lambda}{\sup}\;\; \frac{\lambda^T\{t-\mu(\varrho)\}}{\{\lambda^T\Sigma(\varrho)\lambda\}^{1/2}},
	\end{align}where $\Lambda$ is some subset of $\mathbb{R}^K$ without the zero vector. The adversary may be thought of as embodying future counterclaims regarding the study's conclusions. In keeping with the scientific method the investigator recognizes that her conclusions will be subjected to challenges by her peers, and through the sensitivity analysis assesses whether a particular counterclaim could possibly overturn the study's findings.  The critic aligns the unobserved confounders to inflate the $p$-value for the performed inference, while the investigator may choose weights for each outcome within the constraints imposed by $\Lambda$ in response to the configuration of unmeasured confounders selected by the critic. \textcolor{black}{With regards to \eqref{eq:twoplayer}, as $\Gamma$ grows during the process of a sensitivity analysis the outer minimization takes place over a sequence of growing feasible regions. In the sense of the two-player game, this corresponds to the adversary having more and more flexibility in assigning unfavorable treatment allocation distributions.}  
	
	Most familiar large-sample sensitivity analyses for Fisher's sharp null hypothesis are instances of this game for particular choices of $\Lambda$. Setting $\Lambda = \{e_k\}$ where $e_k$ is a vector with a 1 in the $k$th coordinate and zeroes elsewhere returns a univariate sensitivity analysis for the $k$th outcome with a greater-than alternative, while $-e_k$ would return the less-than alternative. When the test statistics $T_K$ are rank tests, setting $\Lambda = \{1_K\}$ where $1_K$ is a vector containing $K$ ones returns the coherent rank test of \citet{ros97}. When $\Lambda = \{e_1,...,e_K\}$, the collection of standard basis vectors, (\ref{eq:twoplayer}) returns the method of \citet{FS16} with greater-than alternatives, and $\Lambda = \{\pm e_1,..., \pm e_K\}$ gives the same method with two-sided alternatives. The method of \citet{ros16} amounts to a choice of $\Lambda = \mathbb{R}^K\; \backslash \; \{0_K\}$, i.e. all possible linear combinations except the vector $0_K$ containing $K$ zeroes. 
	
	While appealing as a unifying framework for multivariate sensitivity analyses, the form (\ref{eq:twoplayer}) would be of little practical use if the corresponding optimization problem could not be readily solved. The problem (\ref{eq:twoplayer}) is not itself convex; however, consider replacing it with
	\begin{align}\label{eq:mod}
	{b}^*_{\Gamma, \Lambda} = \underset{\varrho\in\cP_\Gamma}{\min}\;\; \underset{\lambda \in \Lambda}{\sup}\;\; \max\left[0,\frac{\lambda^T\{t-\mu(\varrho)\}}{\{\lambda^T\Sigma(\varrho)\lambda\}^{1/2}}\right]^2,
	\end{align}
	and let $f(\lambda, \varrho) = \max[0, \lambda^T\{t-\mu(\varrho)\}/\{\lambda^T\Sigma(\varrho)\lambda\}^{1/2}]^2$. This replaces negative values for the standardized deviate with zero, and then takes the square of the result. It is a monotone non-decreasing transformation of the standardized deviate in general, and is strictly increasing whenever the standardized deviate is larger than zero. The following proposition, proved in the Supplementary Material, establishes convexity of (\ref{eq:mod}).
	\begin{proposition}\label{prop: convexity of objective}
		The function $g(\varrho) = \sup_{\lambda \in \Lambda} f(\lambda,\varrho)$ is convex in $\varrho$ for any set $\Lambda$ without the zero vector.
	\end{proposition} 
	
	The proof requires showing that for any $\lambda\in\Lambda$, the function $f(\lambda, \varrho)$ is convex in $\varrho$. As the pointwise supremum over a potentially infinite set of convex functions is itself convex \citep[][\S 3.2.3]{BV04}, the result then follows. The convexity of $g(\varrho)$ allows for its minimization over the polyhedral set $\cP_\Gamma$ such that the value $b^*_{\Gamma,\Lambda}$ in (\ref{eq:mod}) can be computed in practice.  For any $\Gamma$ and $\Lambda$, a sensitivity analysis through (\ref{eq:twoplayer}) would proceed by comparing the value ${a}^*_{\Gamma,\Lambda}$ to a suitable critical value $c_{\alpha,\Lambda}$.  Observe that ${a}^*_{\Gamma,\Lambda} = (b^*_{\Gamma,\Lambda})^{1/2}$ for ${a}^*_{\Gamma,\Lambda} \geq 0$. If $\alpha \leq 0.5$ then $c_{\alpha,\Lambda}$ is non-negative for any choice of $\Lambda$. Consequently ${a}^*_{\Gamma,\Lambda} \geq c_{\alpha, \Lambda}$, leading to a rejection of the null, if and only if $b^*_{\Gamma,\Lambda} \geq	c_{\alpha,\Lambda}^2$ so long as $\alpha \leq 0.5$. Through this equivalence, a large-sample sensitivity analysis using (\ref{eq:twoplayer}) can proceed through the solution of the convex program (\ref{eq:mod}).
	
	\subsection{The practitioner's price}
	The critical value $c_{\alpha,\Lambda}$ depends on the structure of $\Lambda$, through which it is seen that additional flexibility in the set $\Lambda$ does not come without a cost. Intuition for the price to be paid may be formed at $\Gamma=1$ in (\ref{eq:twoplayer}). When $\Lambda$ is a singleton the asymptotic reference distribution is the standard normal. If $\Lambda$ is instead a finite set with $|\Lambda| = L > 1$ simply comparing the optimal value of (\ref{eq:global}) to the $1-\alpha$ quantile of a standard normal would not provide a level$-\alpha$ test due to multiplicity issues. One could proceed using a Bonferroni correction based on the $L$ comparisons, which would inflate the critical value. When $\Lambda = \mathbb{R}^K\;\backslash\; \{0_K\}$, \citet{ros16} applies a result on quadratic forms of multivariate normals \setcitestyle{square}\citep[e.g.][page 60, 1f.1(i)]{Rao73}\setcitestyle{round} to show that one must instead use the square root of a  critical value from a $\chi^2_K$ distribution when conducting inference through (\ref{eq:twoplayer}). This result underpins Scheff\'e's method for multiplicity control while comparing all linear contrasts of a multivariate normal \citep{sch53}. In the potential presence of hidden bias, the additional flexibility afforded by a richer set $\Lambda$ often offsets the loss in power from controlling for multiple comparisons, particularly in large samples. We discuss this further in \S \ref{sec: Design Sensitivity}, but see also \citet[][\S 6]{FS16} and \citet[][\S 4]{ros16}.

	\section{The Null Distribution Over Coherent Combinations}
	\subsection{Adaptive linear combinations over the non-negative orthant} \label{sec: Nonnegative Orthant}
	By allowing the set $\Lambda$ to be arbitrary, the developments \S \ref{sec:twoplayer} were presented with Fisher's sharp null in mind. A moment's reflection reveals that should inference instead concern the composite null (\ref{eq:global}) of non-positive effects for all outcome variables, the set $\Lambda$ must be constrained to maintain the desired size of the procedure. If $\Lambda$ allows for arbitrary linear combinations, evidence consistent with non-positive treatment effects for each outcome variable may nonetheless result in a rejection of the null hypothesis based on (\ref{eq:twoplayer}) beyond the nominal rate by setting $\lambda_k$ negative for each $k$. Directional control is lost without constraining the signs of the elements of $\Lambda$.
	
	Following \citet[][\S 9.4]{R02}, we define a family of coherent test statistics by restricting the  vector $\lambda$ to lie in the non-negative orthant, $\Lambda_+ = \{\lambda: \lambda_k \geq 0\;\; (k=1...,K);\;\; \sum \lambda_k > 0\}$. The coherent test of \citet{ros97} with $\lambda = 1_K$ is a particular element of $\Lambda_+$. We instead consider a large-sample sensitivity analysis for (\ref{eq:twoplayer}) with $\Lambda=\Lambda_+$, hence optimizing over the entire space of coherent linear combinations. We describe a projected subgradient descent method for solving (\ref{eq:mod}) with $\Lambda=\Lambda_+$ in the Supplementary Material. Subgradients are straightforward to compute, and projections onto $\cP_\Gamma$ are facilitated by the constraints being separable across matched sets. 
	
	Let $\tilde{\varrho}$ be the true, though typically unknown, vector of assignment probabilities and consider the random variable
	\begin{align}\label{eq:A}
	A_{\Lambda_+}(Z,R) = \underset{\lambda \in \Lambda_+}{\sup}\;\; \frac{\lambda^T\{T-\mu(\tilde{\varrho})\}}{\{\lambda^T\Sigma(\tilde{\varrho})\lambda\}^{1/2}}.
	\end{align}
	Let $R_Z$ denote the observed responses when the treatment assignment is $Z$. Let $G(v, R_Z)$ be the reference distribution based on the observed outcome $R_Z$ assuming Fisher's sharp null,
	\begin{align}\label{eq:randdist}
	G(v, R_Z) &= \sum_{b\in \Omega}1\{A_{\Lambda_+}(b, R_Z)\leq v\}\P(Z=b\mid \cF, \cZ),
	\end{align} and let $G^{-1}(1-\alpha, R_Z)$ be its $1-\alpha$ quantile. Observe that the reference distribution $G(v, R_Z)$ itself varies over elements of $\Omega$ through its dependence on $R_Z$ if Fisher's sharp null is false.
	
	Proposition \ref{prop:size}, proved in the Supplementary Material, states that a valid test of the composite null of non-positive effects $H_0$ can be achieved through the randomization distribution of $A_{\Lambda_+}$ under the assumption of Fisher's sharp null. Through an analogous proof, the randomization distribution also provides an unbiased test against positive alternatives of the form $\tau_{ijk} \geq 0$ $(i=1,..,I; j=1,...,n_i; k = 1,...,K)$ with at least one strict inequality.
	
	\begin{proposition}\label{prop:size}
		Suppose that the global null (\ref{eq:global}) of non-positive treatment effects is true and assume that the test statistics $T_k$ $(k=1,...,K)$ are effect increasing.  Then
		\begin{align*}
		\P\{A_{\Lambda_+}(Z,R_Z) \geq G^{-1}(1-\alpha, R_Z)\} \leq \alpha,
		\end{align*}
		such that the reference distribution under Fisher's sharp null controls the Type I error rate for any element of the composite null $H_0$. 
	\end{proposition}

	Both the observed value $A_{\Lambda_+}=a_{\Lambda_+}$ and the probabilities $\P(Z=z\mid \cF, \cZ)$ are unknown in the observational study at hand due to their dependence on the true conditional assignment probabilities $\tilde{\varrho}$. Through the solution to (\ref{eq:twoplayer}) we instead observe the value ${a}^*_{\Gamma,\Lambda_+}$, which bounds $a_{\Lambda_+}$ from below so long as $\tilde{\varrho}\in \cP_\Gamma$. That said, the true randomization distribution (\ref{eq:randdist}) typically remains unknown outside of a randomized experiment as it depends on $\tilde{\varrho}$. For many test statistics, such as those formed when $\Lambda$ is a singleton, the asymptotic reference distribution does not depend on $\tilde{\varrho}$ after suitable standardization. In what follows, we consider the large-sample distribution of $A_{\Lambda_+}$ under Fisher's sharp null.
	
	\subsection{The chi-bar-squared distribution}
	Comparing the optimal value of (\ref{eq:twoplayer}) with $\Lambda=\Lambda_+$ to the $1-\alpha$ quantile of a standard normal would not provide a valid level$-\alpha$ sensitivity analysis, as it would not account for the optimization over coherent combinations. While one could proceed with the square root of the $1 - \alpha$ quantile of a $\chi_{K}^{2}$ distribution, doing so would be unduly conservative. The $\chi^2_K$ critical value allows for optimization over all linear combinations, while here we have constrained ourselves to combinations lying in the non-negative orthant. Theorem \ref{thm: asymptotic chi bar sqared} provides the appropriate reference distribution given this restriction.

	\begin{theorem}\label{thm: asymptotic chi bar sqared}
		Suppose that $I^{-1}\Sigma(\tilde{\varrho})$ has an positive definite limit $M$ as $I\rightarrow \infty$ and the random vector $\Sigma(\vect{\tilde{\varrho}})^{-1/2}\left\{T - \mu(\vect{\tilde{\varrho}})\right\}$ converges in distribution to a $K$-dimensional vector of independent standard normals.  Then, as $I \rightarrow \infty$ the random variable $A_{\Lambda_+}^2$ converges in distribution to a $\ChiBarSq(M^{-1}, \Lambda_+)$ random variable under Fisher's sharp null. 
	\end{theorem}
	
	The proof is deferred to the Supplementary Material. The Supplementary Material also contains a discussion of sufficient conditions such that $\Sigma(\tilde{\varrho})^{-1/2}\{T - \mu(\tilde{\varrho})\}$ converges in distribution to a multivariate normal, which amount to assumptions about the vectors of constants $q_k$ $(k=1,...,K)$. For instance, one sufficient condition would be to stipulate that $I^{-1}\sum_{i = 1}^{I}\sum_{j = 1}^{n_{i}}q_{ijk}^{4}$ is uniformly bounded for all $I \in \N$ and all $k = 1, \ldots, K$.

	The  $\ChiBarSq$  (``chi-bar-squared") is a common family of distributions arising in order restricted statistical inference \textcolor{black}{\citep{SS02}}.  To illustrate, let $X$ be a mean zero $K$-variate normal random vector with positive definite covariance matrix $V$,  and define the random variable
	\begin{equation}\label{eqn: the chi-bar-sq random variable}
	\ChiBarSq(V, \Lambda_+) = X^{T}V^{-1}X - \inf\limits_{\theta \in \Lambda_+}(X - \theta)^{T}V^{-1}(X - \theta).
	\end{equation}
	Letting $\theta$ denote the mean vector of a multivariate normal, (\ref{eqn: the chi-bar-sq random variable}) is equivalent to the likelihood ratio statistic for testing the null $H_0: \theta_k = 0$ $(k=1,...,K)$ versus the alternative $H_a: \theta_k  \geq 0$ $(k=1,...,K)$ with strict inequality in at least one component \citep{K63}. Observe that replacing $\Lambda_+$ with $\mathbb{R}^K$ in (\ref{eqn: the chi-bar-sq random variable}) would return $X^TV^{-1}X$, and with it the usual $\chi^2_K$ distribution. Computation of (\ref{eqn: the chi-bar-sq random variable}) requires solving a quadratic program, an easy task with modern solvers but one which historically limited the adoption of methods requiring the $\ChiBarSq$ distribution.
	
	The cumulative distribution function of the $\ChiBarSq(V, \Lambda_+)$ is $\P\{\ChiBarSq(V, \Lambda_+) \leq c\} = \sum_{i = 0}^{K}w_{i}(V, \Lambda_+)\Prob{\chi^{2}_{i} \leq c},$ a mixture of $\chi^2_i$ distributions $(i=0,...,K)$ with $\chi^2_0$ representing a pointmass at zero. The $i$th weight $w_{i}(V, \Lambda_+)$ is equal to the probability that the vector $V^{-1/2}X$ has exactly $i$ positive components. The weights depend upon the covariance $V$ through the corresponding correlation matrix $C$: any two covariance matrices $V'$ and $V$ with the same correlation structure $C$ yield the same weights for $\bar{\chi}^2$ \setcitestyle{square}\citep[][Proposition 3.6.1 (11)]{SS05}. \setcitestyle{round} See \citet{K63, rob88order}; and \citet{SS05} for more on the role of the $\ChiBarSq$ distribution in multivariate one-sided testing.
	
	\citet{S03} presents an extension of Scheff\'e's method for multiple comparisons to linear combinations subject to cone constraints such as lying in the non-negative orthant. Arguments therein show that strong duality holds in (\ref{eq:A}), such that the optimal value for (\ref{eq:A}), $A_{\Lambda_+}$, equals the optimal value of the dual. The optimal solution to the dual is
	\begin{equation}\label{eqn: heuristic proof inner objective}
	A_{\Lambda_+} = \left\{h^{T}\Sigma(\tilde{\varrho})h - \inf\limits_{\lambda \in \Lambda_+}(h - \lambda)^{T}\Sigma(\tilde{\varrho})(h - \lambda)\right\}^{1/2},
	\end{equation}
	where $h = \Sigma^{-1}(\tilde{\varrho})\{T-\mu(\tilde{\varrho})\}$. Under mild conditions $h$ is asymptotically multivariate normal with covariance equal to the limit of $I\Sigma^{-1}(\tilde{\varrho})$.  Comparing (\ref{eqn: heuristic proof inner objective}) to (\ref{eqn: the chi-bar-sq random variable}) provides intuition for the $\bar{\chi}^2$ limiting distribution. Moving forwards, we refer to the procedure using $A_{\Lambda+}$ to facilitate inference as the $\ChiBarSq$-test. In the Supplementary Material, we present Type I error control simulations indicating that the $\ChiBarSq$ reference distribution provides a reasonable approximation to the true randomization distribution of $A_{\Lambda+}$ with moderate sample sizes.
	
	\subsection{The critical value and its dependence on the unknown assignment probabilities}   \label{sec:critical value}
	A large-sample sensitivity analysis can be conducted by comparing the optimal value of (\ref{eq:twoplayer}) over coherent linear combinations, $a_{\Gamma, \Lambda_+}$ to the square root of the $1 - \alpha$ quantile of a $\ChiBarSq\{\Sigma^{-1}(\tilde{\varrho}), \Lambda_+\}$ distribution. Recalling that $\tilde{\varrho}$ is the true vector of assignment probabilities, we are faced with a difficulty encountered by neither a univariate sensitivity analysis for a particular outcome nor the method of \citet{ros16}: The asymptotic reference distribution depends on the assignment probabilities $\tilde{\varrho}$ through the covariance $\Sigma(\tilde{\varrho})$ even after proper normalization. While $\tilde{\varrho}$ is known in a randomized experiment, the purpose of a sensitivity analysis is to assess robustness of a study's findings as $\tilde{\varrho}$ is allow to vary within bounds imposed by $\Gamma$.
	
	The dependence of the covariance on nuisance parameters is commonly encountered in applications of the $\ChiBarSq$ distribution \citep[][\S 2.2]{SS02}. One solution is to compute $p$-values through the bound $\P\{\ChiBarSq(V, \Lambda_+) \geq c\} \leq 0.5 \{\Prob{\chi_{K - 1}^{2} \geq c} + \Prob{\chi_{K}^{2} \geq c}\}$; see \citet[][Theorem 6.2]{per69} for a proof. This upper bound is attained in the limit as the correlation between all outcomes converges to one, and can itself be quite conservative in the presence of more moderate degrees of correlation typically observed in practical applications.
	
	Motivated by the particular structure imposed by a sensitivity analysis, we instead use a two-stage procedure to better upper bound the worst-case critical value for each $\Gamma$. In a sensitivity analysis the range of the nuisance parameters $\tilde{\varrho}$ is controlled by $\Gamma$. At $\Gamma=1$ $\tilde{\varrho}$ is entirely specified, such that in finely stratified experiments the appropriate $\ChiBarSq$ distribution is known. As $\Gamma$ increases the bounds imposed by membership in $\cP_\Gamma$ widen. For each pair of outcomes $k$ and $\ell$, we first find upper and lower bounds on the correlation between $k$ and $k'$ given $\tilde{\varrho} \in \cP_\Gamma$, call them $C^{(\ell)}_{k,k',\Gamma}$ and $C^{(u)}_{k,k',\Gamma}$. We then maximize the $1-\alpha$ quantile of a $\ChiBarSq(C^{-1},\Lambda_+)$ distribution over the correlation matrix $C$ subject to $C^{(\ell)}_{k,k',\Gamma}\leq C_{k,k'} \leq C^{(u)}_{k,k',\Gamma}$ for all $k,k'$ and $C$ being a correlation matrix. See the Supplementary Material for implementation details along with a discussion of the case $K=2$, where it is seen that the worst-case critical value is attained at the lower bound on the correlation. In practice, we find that this can provide meaningful improvements in the power of the procedure; see the Supplementary Material for an illustration.
	
	\section{Design Sensitivity and Power for the $\ChiBarSq$-Test}
	\subsection{Design sensitivity}\label{sec: Design Sensitivity}
	Suppose that the treatment in question actually has an effect in the direction of the alternative, and further that there is truly no hidden bias such that inference at $\Gamma=1$ would be justified. As would be the case in practice, the researcher analyzing the observational study is unaware of these favorable conditions. Thus, she would like to reject the null hypothesis not only under the assumption of no unmeasured confounding, but also for values $\Gamma>1$ to assess whether the rejection of the null is robust to certain degrees of hidden bias. The power of a level-$\alpha$ sensitivity analysis is the probability that the procedure correctly rejects the null hypothesis at some pre-specified value of $\Gamma \geq 1$. In what follows we will assume a stochastic generative model for the outcome variables, an assumption which greatly simplifies power calculations.
	
	Under mild conditions, there is a value $\designSens$ such that the power of a sensitivity analysis converges to one for all $\Gamma < \designSens$, and converges to zero for all for all $\Gamma > \designSens$; this value is called the design sensitivity of the test \citep{ros04}. It quantifies the asymptotic ability of the test to discriminate treatment effect under the concern of bias in the treatment allocation process, and can vary substantially across choices of test statistics. For a fixed data generating model, a test with high design sensitivity is preferable to a test with low design sensitivity.
	
	For fixed choices of the univariate test statistics $T_k = Z^Tq_k$ $(k=1,...,K)$, we consider the design sensitivity of multivariate tests based upon (\ref{eq:twoplayer}) and their dependence on the set $\Lambda$. Theorem~\ref{thm: nested lambda design sensitivity} shows that design sensitivity is a monotonic non-decreasing function with respect to the partial ordering over sets $\Lambda$ given by inclusion.
	\begin{theorem}\label{thm: nested lambda design sensitivity}
		Suppose $\Lambda_{1} \subseteq \Lambda_{2}$. Under mild conditions, the design sensitivity of \eqref{eq:twoplayer} using $\Lambda = \Lambda_{1}$ is less than or equal to the design sensitivity of \eqref{eq:twoplayer} using $\Lambda = \Lambda_{2}$.
	\end{theorem}
	The proof of Theorem~\ref{thm: nested lambda design sensitivity} is deferred to the Supplementary Material. In light of Theorem~\ref{thm: nested lambda design sensitivity}, it may be tempting to take $\Lambda = \R^{K} \setminus \{0_K\}$ in order to achieve the greatest design sensitivity. While this would result in a valid test of Fisher's sharp null, it does not provide a valid test of the null hypothesis of non-positive treatment effects: should the signs of $\lambda_k$ be left unconstrained, evidence of a negative treatment effect may result in a large optimal value for (\ref{eq:twoplayer}). Restricting attention to the set of coherent linear combinations $\Lambda_+$, Theorem~\ref{thm: nested lambda design sensitivity} gives rise to the following optimality property for the $\ChiBarSq$-test due to its optimizing over the entirety of $\Lambda_{+}$.
	\begin{corollary}\label{cor: optimality over coherent tests}
		The $\ChiBarSq$-test achieves greatest design sensitivity among coherent tests based upon \eqref{eq:twoplayer} with $\Lambda\subseteq \Lambda_+$
	\end{corollary}

	\subsection{Finite-sample power for rejecting the global null}\label{sec:global null}\label{sec: Power Simulations}
	Corollary \ref{cor: optimality over coherent tests} illustrates that despite the larger critical value necessitated by the $\ChiBarSq$-test by optimizing over $\Lambda_+$,  the $\ChiBarSq$-tests achieves the largest possible design sensitivity over the set of coherent multivariate tests. This reflects that in large samples bias trumps variance in the analysis of observational studies, such that the differences in critical values are rendered irrelevant in the limit. In moderate samples, the variance of the null distribution plays a larger role in the power of a sensitivity analysis, such that differences in critical values can make a more substantial difference for procedures with similar design sensitivities. 
	
	We present a simulation study comparing the power of a sensitivity analysis based upon the $\ChiBarSq$-test to two competitors: the method of \citet{FS16}; and the test using \eqref{eq:twoplayer} with $\Lambda = \{1_K\}$, which we refer to as the equal-weight test. Combining test statistics with equal weights is only sensible when the constituent test statistics $T_k$ $(k=1,..,K)$ reflect evidence against the null hypothesis on the same scale. This would be true of rank statistics as described in \citet{ros97}, and would also be true of suitably scaled $m$-statistics of the type described in \citet{ros07}; however, if one outcome is tested using a rank-sum statistic and another with an $m$-statistic for instance, the ``equal-weight" test would give unreasonable weight to the rank-sum recorded outcome.  The $\ChiBarSq$-test and the test of \citet{FS16} do not require comparable scales for the test statistics as they are scale invariant.
	
	The simulations are performed on $I=300$ matched pairs with $K=3$ outcomes. In each simulation, we generate $I$ mean-zero unit-variance trivariate normal vectors of noise $(\varepsilon_{i1}, \varepsilon_{i2}, \varepsilon_{i3})^{T}$ equicorrelated with correlation $\rho$. We then create the vector of treated-minus-control paired differences in outcomes as $(Y_{i1}, Y_{i2}, Y_{i3})^{T} = (\tau_{1}, \tau_{2}, \tau_{3})^{T} + (\varepsilon_{i1}, \varepsilon_{i2}, \varepsilon_{i3})^{T}$ for different values of the treatment effects $(\tau_{1}, \tau_{2}, \tau_{3})^{T}$. For each outcome variable, the employed test statistic is $T_k = \sum_{i=1}^I\sign(Y_{ik})\min(|Y_{ik}|/s_k, 2.5)$, where $s_k$ is the median of $|Y_{ik}|$ $(i=1,..,I)$. This amounts to a choice of a $m$-statistic with Huber's $\psi$-function, as described in \citet{ros07}.

	Table \ref{tab: design sens for chibarsq, equal-weight, and FS16} presents the values of the treatment effects and the correlation employed in the simulation study. For each combination of parameters, it further provides the design sensitivity for the $\ChiBarSq$-test and the equal-weight test. While there is no known formula for the design sensitivity of the procedure of \citet{FS16}, it is lower-bounded by the the maximal design sensitivities of the three univariate tests; this value is also presented in the table. The table reflects Corollary \ref{cor: optimality over coherent tests}: for each combination of parameters, the design sensitivity for the $\ChiBarSq$-test is greater than or equal to that of the equal-weight test and the maximal univariate test. Further, there is no consistent ordering between the equal-weight test and the max of the univariate tests, as the corresponding sets $\Lambda$ for neither test is a subset of the other.

	\begin{table}
		\centering
		\begin{tabular}{l cc cc cc}
			& \multicolumn{2}{c}{$\ChiBarSq$-Test}          & \multicolumn{2}{c}{Equal-Weight Test}  & \multicolumn{2}{c}{Max Univariate}\\
			& $\rho = 0$   & $\rho = 0.2$                    & $\rho = 0$      & $\rho = 0.2$       & $\rho = 0$      & $\rho = 0.2$\\ 
			
			$\tau = (0.25, 0.25, 0.25)^{T}$                & 2.9          & 2.4                            & 2.9             & 2.4              & 1.9                & 1.9\\
			$\tau = (0.10, 0.10, 0.50)^{T}$                   & 3.6          & 3.4                            & 2.6             & 2.2              & 3.4                & 3.4\\
			$\tau = (0.02, 0.20, 0.50)^{T}$                  & 3.8          & 3.5                            & 2.8             & 2.4              & 3.4                & 3.4
		\end{tabular}
		
		\caption{Design sensitivities for $\ChiBarSq$-test, the equal-weight test, and the largest of the three univariate tests under both independence and moderate positive correlation between outcomes. }
		\label{tab: design sens for chibarsq, equal-weight, and FS16}
	\end{table}

	Figure~\ref{fig: Normal Huber power comparisons} presents the estimated power curves of the three tests as a function of $\Gamma > 1$ in these simulation settings at $I=300$, with 2000 simulations for each combination of parameters. The correlation between paired differences varies across the columns from $\rho=0$ (left) to $\rho = 0.2$ (right), while the treatment effects vary down the rows. The first row corresponds to $\tau_{1} = \tau_{2} = \tau_{3}$ and $I$, and here it is seen that  the equal-weight test outperforms both the $\ChiBarSq$-test and \citet{FS16}. When the treatment effects are equal the linear combination $\lambda=1_K$ attains the largest design sensitivity, and by restricting $\Lambda$ to only this linear combination the lower critical value employed by the equal-weighted test improves power over that attained by the $\ChiBarSq$-test.  When one of the three outcomes is strongly affected by the treatment while the other two are minimally impacted, as in the second row of the figure, the method of \citet{FS16} and the $\ChiBarSq$-test perform similarly, while the equal-weight test lags behind. The test statistics returned by the $\ChiBarSq$-test are larger, but this is offset relative to the method of \citet{FS16} by the larger critical value necessitated. When the treatment effects are staggered between the three outcomes as in the third row, the $\ChiBarSq$-test outperforms both \citet{FS16} and the equal-weight test, particularly in the case of independence between the outcome variables. Optimizing over $\Lambda_{+}$ increases the value of the test statistic over both competitors, such that the flexibility is well worth the price of a larger critical value.

	\begin{figure}
		\centering
		\includegraphics[scale=.65]{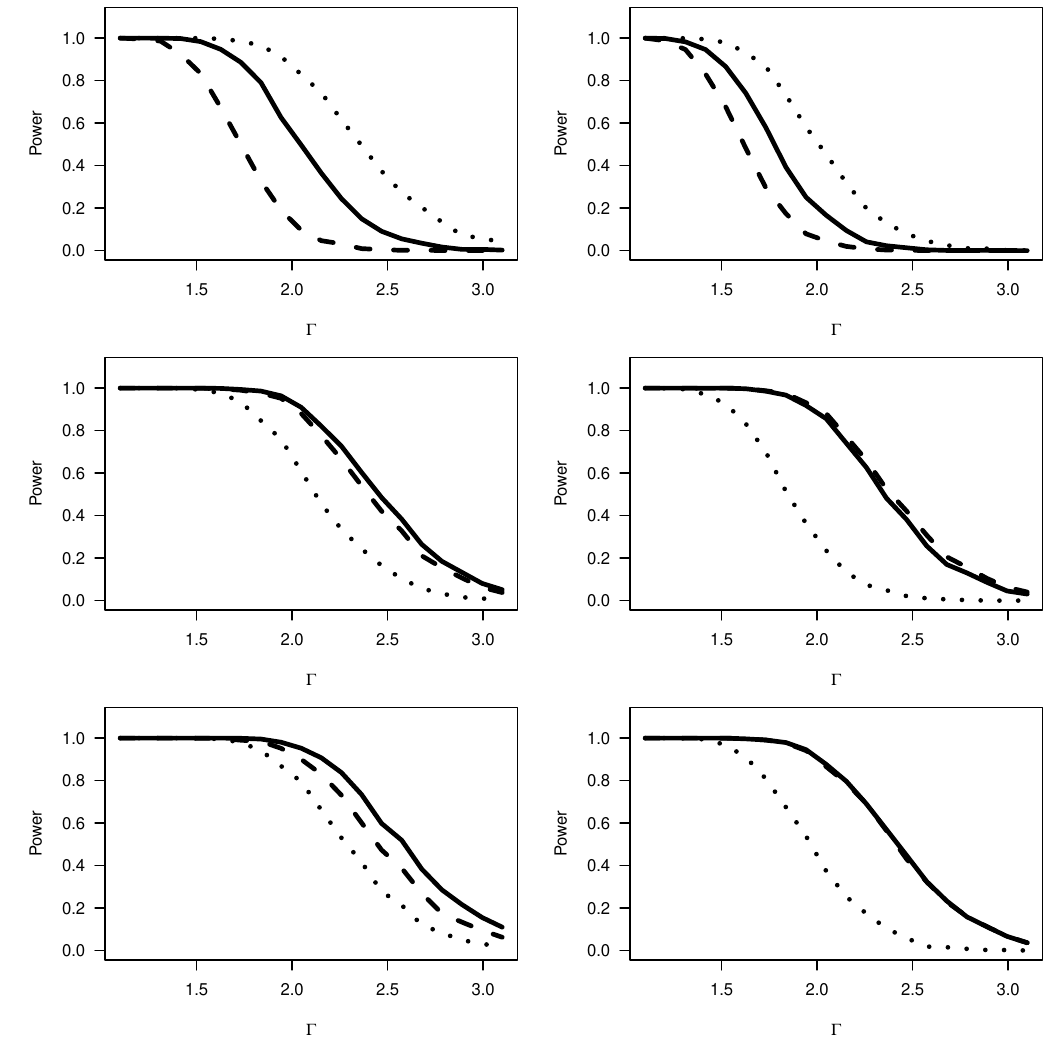}
		\caption{Power comparisons between the method of \citet{FS16} (dashed), the $\ChiBarSq$-test of this paper (solid), and the equal-weight test (dotted) as $\Gamma$ increases with $I = 300$. The first row has $\tau_{1} = \tau_{2} = \tau_{3} = 0.25$; the second row has $\tau_{1} = \tau_{2} =0.1$ and $\tau_{3} = 0.5$; and the third row has $\tau_{1} = 0.05$,  $\tau_{2} =0.2$, and $\tau_{3} = 0.5$.  The left column has $\rho = -0.2$, the center has $\rho = 0$, and the right has $\rho = 0.2$. }
		\label{fig: Normal Huber power comparisons}
	\end{figure}

	The simulations indicate that while the $\ChiBarSq$-test must have optimal power in the limit as asserted by Corollary \ref{cor: optimality over coherent tests}, it need not have the best finite-sample performance. In some cases the equal-weight test can outperform it, while in others it is outperformed by the method of \citet{FS16}. Importantly the $\ChiBarSq$-test was never the worst of the three methods considered, and the simulations show that \textit{a priori} restricting the set of combinations $\Lambda$ under consideration can substantially reduce power should the choice of $\Lambda$ be poor. For instance, the equally-weighted test performs poorly in the second and third rows of Figure \ref{fig: Normal Huber power comparisons}, while the method of \citet{FS16} is markedly worse than the other methods in the first row. The $\ChiBarSq$-test does pay a price in terms of an increased critical value, but this price acts as insurance against an unwise choice of $\Lambda$.  \textcolor{black}{Theorem~\ref{thm: nested lambda design sensitivity} offers asymptotic assurance that the $\ChiBarSq$-test performs optimally in terms of design sensitivity; furthermore, the results of Table~\ref{tab: design sens for chibarsq, equal-weight, and FS16} and Figure~\ref{fig: Normal Huber power comparisons} demonstrate that the $\ChiBarSq$-test  performs well across a broad range of treatment effect regimes without sacrificing asymptotic optimality.} 
	
	In the Supplementary Material, we present additional simulations with $I=1000$ matched pairs which begin to show convergence of behavior of the tests under comparison to their design sensitivities. We further illustrate the potential for improvements in power for testing outcome-specific null hypotheses through incorporating the $\ChiBarSq$-test into a closed testing framework, as described in \citet[][\S 6]{FS16}.
	
	\section{Illustrations of Multivariate One-Sided Sensitivity Analysis}
	
	\label{sec: Data Examples}
	\subsection{The role of coherence in two observational studies}
	We now consider the role of multiple outcomes in two observational studies.  Both examples are drawn from The National Health and Nutrition Examination Survey (NHANES) and study physiological impacts of cigarette smoking. One study investigates the impact of smoking on two measures of periodontal disease, while the other looks at whether smoking increases urinary metabolite levels of four carcinogens.  In both examples, the alternative hypothesis is that smoking should have a positive treatment effect on each of the outcome variables measured.  Rosenbaum remarks that ``If incoherence presents a substantial obstacle to a claim that the treatment caused its ostensible effects, then the absence of incoherence - that is, coherence - should entail some strengthening of that claim" \citep[][p. 119]{designofobs}. Should the evidence suggest ostensible effects of smoking incompatible with positive effects for each outcome variable, smoking's place in the causal pathway would be cast into doubt. Should the outcomes all be affected in the predicted direction, this would provide further evidence for smoking's role in the causal mechanism. 
	
	In both observational studies and for each outcome variable, we use an $m$-test based upon Huber's $\psi$-function to conduct inference with the default choices for parameters in the \texttt{senmv} function in the \texttt{sensitivitymv} package in \texttt{R}. 
	\subsection{Smoking and periodontal disease}\label{sec:perio}    
	It has been suggested that up to 42\% of cases of periodontal disease can be attributed to smoking \citep{TA00}; however, as the evidence is observational in nature this association may well be explained away by other intrinsic differences between smokers and non-smokers.  Using the 2011-2012 NHANES survey, \citet{ros16} paired $I=441$ smoking individuals to non-smokers who were similar on the basis of education, income, race, age and gender. Two outcome variables pertaining to dental health were recorded, one each for upper and lower teeth. In this context, coherence would amount to demonstrating that smoking negatively impacted dental health in both the upper and lower teeth. Such a coherent hypothesis strengthens the causal claim that cigarette smoking is detrimental to periodontal health. Should smoking only appear to impact upper teeth but not lower teeth, for instance, such incoherence would cast into doubt whether smoking is truly to blame.  
	
	At $\alpha=0.05$, the overall null hypothesis of non-positive treatment effects was rejected up until $\Gamma= 2.36$ when using the $\ChiBarSq$-test, while the equal-weight test was able to reject until $\Gamma = 2.54$.   By selecting $\Lambda = \Lambda_{+}$ Theorem~\ref{thm: asymptotic chi bar sqared} gives that the appropriate asymptotic null distribution is the $\ChiBarSq$ distribution, while for the equal-weight test the asymptotic null distribution is the standard normal. The $1-\alpha$ quantile of the standard normal lies below the square root of the $1 - \alpha$ quantile of any $\ChiBarSq$ distribution, such that the equal-weight test is able to employ a smaller critical value.  With this restriction comes the risk that equally weighting the outcome variables may be suboptimal.  In this particular observational study, it comes as little surprise that with periodontal disease in upper and lower teeth the risk was worth the while: there is little reason to suspect that magnitude of effects on upper and lower teeth should differ.  Sensitivity analysis using the method of \citet{FS16} achieves significance up to $\Gamma = 2.32$, a slightly lower value than the $\ChiBarSq$-test.  The method of \citet{FS16} takes $\Lambda$ as the set of standard unit basis vectors for $\R^{2}$ in this case, and does not combine the two related measures of periodontal disease. Despite the method of \citet{FS16} also having a smaller critical value than the $\ChiBarSq$-test, in this example this was offset by the additional flexibility afforded by the $\ChiBarSq$-test in optimizing over $\Lambda_+$.  The sensitivity analysis using the $\ChiBarSq$-test took 13 seconds to complete on a personal laptop with a 2.60GHz processor with 16GB of RAM for this data set.
	
	\subsection{Smoking and polycyclic aromatic hydrocarbons}  \label{sec:PAH}
	
	Polycyclic Aromatic Hydrocarbons (PAHs) are a class of organic compounds formed during incomplete combustion which have been labeled potentially carcinogenic to humans \citep{BGH02}.  We examine urinary concentrations of four different PAH metabolites in 432 smokers and 1206 non-smokers recorded in NHANES 2007-2008. The four metabolites are 1-hydroxyphenanthrene (1-Phen), 3-hydroxyphenanthrene (3-Phen), 1-hydroxypyrene (1-Pyr), and 9-hydroxyfluorene (9-Fluo). Full matching \citep{Han04} was employed to adjust for a host of measured covariates thought to impact one's decision to smoke and one's exposure to PAHs; see the Supplementary Material for additional details. We then proceed with inference assessing whether cigarette use increases urinary concentrations of these four PAH metabolites. As tobacco smoke contains all of these PAHs,  an incoherent result that none or only some urinary concentrations of PAH metabolites are higher in smokers than in non-smokers be discovered would cast into question whether the association between smoking cigarettes and urinary PAH concentrations was actually causal.  At $\alpha=0.05$, a sensitivity analysis using the $\ChiBarSq$-test yielded significance up to $\Gamma= 6.28$, whereas for the equal-weight test with $\Lambda = \{1_K\}$ the sensitivity analysis was only able to reject up to $\Gamma = 5.38$.  Despite the smaller critical value, in this case restricting oneself to equal weighting led to a markedly lower changepoint value of $\Gamma$ than did the $\ChiBarSq$-test. The method of \citet{FS16} rejected until $\Gamma = 6.18$.  
	
	At $\Gamma=6.28$, our procedure for upper bounding the worst-case critical value for the $\ChiBarSq$-test as described in \S \ref{sec:critical value} returns a bound of 2.20 for the test based upon ${a}^*_{6.18, \Lambda_+}$ in (\ref{eq:twoplayer}). To illustrate the improvements from this approach, the square root of the 0.95 quantile of a $\chi^2_4$ is $3.08$, while employing the conservative bound from \citet[Theorem 6.2]{per69} yields a critical value of 2.96.  The $\ChiBarSq$ sensitivity analysis ran in about 20 minutes on a personal laptop with a 2.60GHz processor with 16GB of RAM.  The length of runtime is dependent upon several factors including the number of strata, the size of the strata, the number of outcome variables, and the number of values of $\Gamma$ tested in the sensitivity analysis.  
	
	\subsection{Improvements in tests of individual null hypotheses} 
	Rejecting the global null hypothesis confirms to the experimenter that at least one of the outcome variables is impacted by treatment in the direction of the alternative.  However in order to appraise a coherent pattern of treatment impact an experimenter will need to examine the local null hypotheses of treatment impact upon each of the outcomes individually.  Correcting for multiple comparisons can be facilitated through many techniques; here we juxtapose embedding the $\ChiBarSq$-test into a closed testing framework against performing $K$ individual sensitivity analyses, one for each outcome variable, while employing a Bonferroni correction.  
	
	

	\begin{table}[h]
		\centering
		\begin{tabular}{ccccccc}
			& \multicolumn{2}{c}{Periodontal Disease} & \multicolumn{4}{c}{Polycyclic Aromatic Hydrocarbons} \\
			& Lower Teeth         & Upper Teeth        & 1-Phen       & 3-Phen      & 1-Pyr      & 9-Fluo      \\
			Closed Testing & 2.26                & 1.82               & 2.13         & 5.28        & 5.25       & 5.78        \\
			Bonferroni     & 2.17                & 1.76               & 1.99         & 4.88        & 4.84       & 5.31        \\
			Uncorrected    & 2.26                & 1.82               & 2.13         & 5.28        & 5.25       & 5.78       
		\end{tabular}
		\caption{Comparison of the closed test changepoint $\Gamma$ versus Bonferroni corrected sensitivity analysis changepoint $\Gamma$ for the data examples at $\alpha=0.05$.  The last row is the benchmark given by conducting individual tests at $\alpha$ without correction for multiplicity.} \label{tab: individ}
	\end{table}      
	Table~\ref{tab: individ} details the changepoint $\Gamma$ values for each individual outcome of the periodontal data and the PAH data while controlling the familywise error rate at $\alpha=0.05$. The $\ChiBarSq$-test embedded into a closed testing framework outperformed the Bonferroni corrected tests for each outcome. Table \ref{tab: individ} also includes the changepoint $\Gamma$ values returned by the univariate sensitivity analyses without a Bonferroni correction, i.e. with each outcome tested at $\alpha=0.05$. The table reveals that through embedding the $\ChiBarSq$-test in a closed testing procedure, in both studies we are able to report the same robustness to unmeasured confounding that would have been attained had we not controlled for multiple comparisons in the first place. Due to the improvements in power along the closed testing path furnished by the $\ChiBarSq$-test, there is no cost for evaluating coherence of all outcome variables relative to the best univariate outcome analysis.

	\section{Discussion}\label{sec: Discussion}
	\textcolor{black}{While we have tailored our presentation to continuous outcome variables, our test is equally applicable with binary outcomes and ordinal outcomes. In fact, potential outcomes of any partially ordered set are amenable to this composite null, and the remaining proofs of this paper hold true so long as the test statistics considered are effect increasing. See \citet[][ \S 2.8.5]{R02} for more on effect increasing statistics for partially ordered outcomes. The composite null $H_{k}$ for the $k$th outcome variable requires an ordered structure to the potential outcomes, and since Proposition~\ref{prop:size} relies only upon effect-increasingness of the test statistic $T_k(\cdot, \cdot)$, the result remains valid as long as one has a suitable partial ordering for the values of the potential outcomes.}
	
	\textcolor{black}{The $\ChiBarSq$-test we develop is not immediately applicable to testing Neyman's weak null.  Interestingly, even assuming strong ignorability as would be the case in a randomized experiment, it is possible for the Type I error rate to exceed $\alpha$ under the weak null. The procedure we present uses a critical value from the asymptotic form of a randomization distribution assuming the sharp null as the sharp null attains the supremum $p$-value over $H_0$ in (\ref{eq:global}). If instead only Neyman's weak null is true for all $K$ outcomes but $H_0$ is not it is possible that unspecified effect heterogeneity would cause the reference distribution used by our procedure to not stochastically dominate the randomization distribution, leading to an invalid procedure.  Unlike the univariate case and the  multivariate case with two-sided alternatives, a simple studentization does \textit{not} fix the problem even asymptotically, as the studentized reference distribution depends upon the correlation between the outcome variables. This parallels known results for multivariate permutation tests conducted in the absence of a group invariance assumption \citep{CR16}. An ongoing area of the authors' research is examining the extent to which bootstrap prepivoting may be used to create a test that both exact under $H_0$ and asymptotically valid for Neyman's weak null at $\Gamma=1$, but as of yet no extension to cases of potential unmeasured confounding has been developed. The extension of sensitivity analyses to such contexts remains an interesting and important open question.}
	
	\textcolor{black}{Our use of the $\ChiBarSq$-test in conjunction with closed testing provides a sensitivity analysis for testing patterns of directed effect among a moderate number of outcomes, as is common in many public health, econometric, and policy applications.  Unfortunately, the combinatorial blow-up inherent to closed testing prohibits large-scale multiplicity control of the sort required for applications to data sets of the scale encountered in genome-wide association studies.  Even in regimes for which closed testing is computationally infeasible, the interpretation of sensitivity analyses as two-player games lends meaningful intuition and will hopefully stimulate further algorithmic development.}

	\section*{Acknowledgements}
	The authors thank the editor, the associate editor, and two reviewers for their comments and suggestions, which substantially improved the article's content and presentation.
	
	\section*{Supplementary Material} Supplementary Material available at \textit{Biometrika} online contains theoretical results,simulation studies, further algorithmic details, additional insight into the $\ChiBarSq$ distribution, further information on the observational study on smoking and polycyclic aromatic hydrocarbons, and an \texttt{R} script for implementing the method proposed in this work.
	
	\bibliographystyle{plainnat}
	\bibliography{bibliography}

	
	
	
	
	\begin{center}
		\huge{Supplementary Materials}
	\end{center}
	
	\section{Proof of Main Results}
	\subsection{Proposition~\ref{prop: convexity of objective}}
	\setcounter{proposition}{0}
	\setcounter{theorem}{0}
	\setcounter{lemma}{0}
	
	\begin{proposition}
		The function $g(\varrho) = \sup_{\lambda \in \Lambda} f(\lambda,\varrho)$ is convex in $\varrho$ for any set $\Lambda$ without the zero vector.
	\end{proposition} 
	
	In order to show that \eqref{eq:mod} is convex in $\vect{\varrho}$ we first prove a lemma.
	
	\begin{lemma}\label{lem: concavity of denominator}
		For a fixed $\lambda \in \R^{K}$ the function $d(\vect{\varrho}) = \lambda^{T}\Sigma(\vect{\varrho})\lambda$ is a concave function of $\vect{\varrho}$.
	\end{lemma}
	\begin{proof}[Proof of Lemma~\ref{lem: concavity of denominator}]
		Define $Q_{i}$ to be the $K$-by-$n_{i}$ matrix where the $(k,j)$th entry is $q_{ijk}$.  Then the Hessian matrix of $d(\vect{\varrho})$ with respect to the variables in the $i$th strata is
		\begin{equation*}
		\nabla_{\varrho_{ij};\; (j = 1, \ldots, n_{i})}^{2} d(\vect{\varrho}) = \frac{-1}{2}Q_{i}^{T}\lambda\lambda^{T}Q_{i}.
		\end{equation*}
		This is negative semi-definite.  By independence between strata, the full Hessian $\nabla_{\vect{\varrho}}^{2} f(\vect{\varrho})$ is the direct sum of the Hessians associated to each stratum.  Thus, the full Hessian matrix is a block diagonal matrix wherein each block is negative semi-definite.  Since the eigenvalues of a block diagonal matrix are the collection of eigenvalues of its constituent blocks, we have that the full Hessian must be negative semi-definite as well.  As a consequence, $d(\vect{\varrho})$ is a concave function of $\vect{\varrho}$.
	\end{proof}
	
	\begin{proof}[Proof of Proposition~\ref{prop: convexity of objective}]
		The identity function $x \mapsto x$ is convex as a function of $x$. Since the point-wise maximum of convex functions is convex $\max\{0, x\}$ is convex as a function of $x$. The quadratic function $a \mapsto a^{2}$ is convex and increasing on the non-negative real line so by \citet[3.10]{BV04}  the function $\psi(x) = \left[\max\{0, x\}\right]^{2}$ is a convex function of $x$.
		
		The perspective of a function $\psi(x)$ is defined to be $\phi(x, v) = v\psi(x / v)$ for $v > 0$; by \citet[3.2.6]{BV04} the perspective of a convex function is convex as well.  Computing the perspective of $\psi$ follows as
		\begin{align*}
		\phi(x, v) &= v\psi(x /v)\\
		&= v\left\{\max(0, x/v)\right\}^{2}\\
		&= v\left\{\frac{\max(0, x)}{v}\right\}^{2}\\
		&= \frac{\max(0, x)^{2}}{v}.
		\end{align*}
		
		Thus, $\phi(x, v) = \max(0, x)^{2} / v$ is convex. Now, consider any fixed $\lambda \geq 0$ and $t$ of dimension $K$.  $\mu(\vect{\varrho})$ is a linear function of $\vect{\varrho}$.  Since affine transformations of linear functions are convex, $\lambda^{T}\{t - \mu(\vect{\varrho})\}$ is convex.  Furthermore, $\lambda^{T}\Sigma(\vect{\varrho})\lambda$ is concave in $\vect{\varrho}$ by Lemma~\ref{lem: concavity of denominator}.  By \citet[3.15]{BV04}, since $\phi(x,v)$ is non-decreasing in $x$ and non-increasing in $v$ the function		\begin{equation*}
		f(\lambda, \vect{\varrho}) = \phi\left(\lambda^{T}\{t - \mu(\vect{\varrho})\}, \lambda^{T}\Sigma(\vect{\varrho})\lambda\right) = \frac{\max\left[0, \lambda^{T}\left\{t - \mu(\vect{\varrho})\right\}\right]^{2}}{ \lambda^{T}\Sigma(\vect{\varrho})\lambda}
		\end{equation*}
		is convex in $\vect{\varrho}$. As $g(\vect{\varrho})$ is the point-wise supremum over all $\lambda \in \Lambda$ of $f(\lambda, \vect{\varrho})$, by \citet[3.7]{BV04} $g(\vect{\varrho})$ is convex in $\vect{\varrho}$ as desired.  The requirement that $\Lambda$ excludes the zero vector ensures that for any positive definite $\Sigma(\vect{\varrho})$ the denominator is always defined and thus $g(\vect{\varrho})$ is defined.
	\end{proof}

	\subsection{Proposition~\ref{prop:size}}
	Here and elsewhere in the supplement, $\Lambda_+$ is once again defined to be the non-negative orthant in $\mathbb{R}^K$ excluding the zero vector, that is $\Lambda_+ = \{\lambda: \lambda_k \geq 0\;\; (k=1...,K);\;\; \sum \lambda_k > 0\}$.
	\setcounter{proposition}{1}
	\begin{proposition}
		Suppose that the global null (\ref{eq:global}) of non-positive treatment effects is true and assume that the test statistics $T_k$ $(k=1,...,K)$ are effect increasing.  Then
		\begin{align*}
		\P\{A_{\Lambda_+}(Z,R_Z) \geq G^{-1}(1-\alpha, R_Z)\} \leq \alpha,
		\end{align*}
		such that the reference distribution under Fisher's sharp null controls the Type I error rate for any element of the composite null $H_0$. 
	\end{proposition} 
	\begin{proof}
		
		\begin{align*}
		&\P\{A_{\Lambda_+}(Z,R_Z) > G^{-1}(1-\alpha, R_Z)\}\\  &=\sum_{z\in \Omega}1\{A_{\Lambda_+}(z, R_z) > G^{-1}(1-\alpha, R_z) \}\P(Z = z\mid \cF, \cZ)\\
		&=\sum_{b\in \Omega}\left[\sum_{z\in \Omega}1\{A_{\Lambda_+}(z, R_z) > G^{-1}(1-\alpha, R_z) \}\P(Z = z\mid \cF, \cZ)\right]\P(Z=b\mid \cF, \cZ)\\
		& \leq \sum_{b\in \Omega}\left[\sum_{z\in \Omega}1\{A_{\Lambda_+}(b, R_z) > G^{-1}(1-\alpha, R_z) \}\P(Z = z\mid \cF, \cZ)\right]\P(Z=b\mid \cF, \cZ)\\
		& =\sum_{z\in \Omega}\left[\sum_{b\in \Omega}1\{A_{\Lambda_+}(b, R_z) > G^{-1}(1-\alpha, R_z) \}\P(Z = b\mid \cF, \cZ)\right]\P(Z=z\mid \cF, \cZ)\\
		& \leq \alpha \sum_{z\in \Omega}\P(Z=z\mid \cF, \cZ) = \alpha.
		\end{align*}
		
		The third line simply multiplies by one in the form of $\sum_{b\in\Omega} \P(Z=b\mid\cF, \cZ)$. The fourth line uses that the test statistics are effect increasing. After rearranging the order of summation in the fifth line, the sixth follows by definition as it simply uses that for any particular $z$, $G^{-1}(1-\alpha, R_z)$ is the $1-\alpha$ quantile corresponding to $G(v, R_z) = \sum_{b\in \Omega}1\{A_{\Lambda_+}(b, R_z) \leq v\}\P(Z = b\mid \cF, \cZ)$.

	\end{proof}
	
	\subsection{Theorem~\ref{thm: asymptotic chi bar sqared}}
	For ease of notation we suppress conditioning on $\cF$ and $\cZ$ when writing expectations and covariances in this section. We again define $T_{k} = \sum_{i = 1}^{I}\sum_{j = 1}^{n_{i}}Z_{ij}q_{ijk}$, and let $\tilde{\varrho}$ represent the true vector of conditional assignment probabilities. For precision quantities such as $\tilde{\varrho}$ should be subscripted by $I$ to denote their dependence on the sample size; this is omitted for improved readibility. 
	
	\begin{theorem}
		Suppose that $I^{-1}\Sigma(\tilde{\varrho})$ has a positive definite limit $M$ as $I\rightarrow \infty$ and the random vector $\Sigma(\vect{\tilde{\varrho}})^{-1/2}\left\{T - \mu(\vect{\tilde{\varrho}})\right\}$ converges in distribution to a $K$-dimensional vector of independent standard normals.  Then, as $I \rightarrow \infty$ the random variable $A_{\Lambda_+}^2$ converges in distribution to a $\ChiBarSq(M^{-1}, \Lambda_+)$ random variable under Fisher's sharp null. 
	\end{theorem}
	
	Before proving Theorem~\ref{thm: asymptotic chi bar sqared}, we establish conditions under which the random vector $\Sigma(\tilde{\varrho})^{-1/2}\left\{T - \mu(\tilde{\varrho})\right\}$ has a multivariate normal limiting distribution.
	
	\begin{lemma}\label{lem: normal distrib}
		Suppose that there exists a $\delta > 0$ for which 
		\begin{equation}\label{eqn: normal distrib criterion}
		\sum_{i = 1}^{I}\Expectation{\mid \sum_{j = 1}^{n_{i}}q_{ijk}Z_{ij} - \sum_{j = 1}^{n_{i}}q_{ijk}\varrho_{ij}^{*} \mid ^{2 + \delta} } = O(I)
		\end{equation}
		for all $k$ and all $I$, and that $I^{-1}\Sigma(\tilde{\varrho})$ has an positive definite limit $M$ as $I\rightarrow \infty$.  Then as $I \rightarrow \infty$ the random vector $\Sigma(\vect{\tilde{\varrho}})^{-1/2}\left\{T - \mu(\vect{\tilde{\varrho}})\right\}$ converges in distribution to a $K$-variate vector of independent standard normals. 
	\end{lemma}
	
	\begin{proof}[Proof of Lemma~\ref{lem: normal distrib}]
		Define $X_{i} = (X_{i1}, \ldots, X_{iK})^{T}$ where $X_{ik} = \sum_{j = 1}^{n_{i}}q_{ijk}Z_{ij}$.   Denote $\mu_{i}(\tilde{\varrho}) = \Expectation{X_{i}}$ and $\Sigma_i(\tilde{\varrho}) = E(X_iX_i^T) - E(X_i)E(X_i)^T$, such that $\sum_{i=1}^I\mu_i(\tilde{\varrho}) = \mu(\tilde{\varrho}) = E(T)$ and $\sum_{i=1}^I\Sigma_i(\tilde{\varrho}) = \Sigma(\tilde{\varrho}) = \text{cov}(T)$.
		
		By the Cram\'{e}r-Wold device it suffices to consider the distribution of the univariate random variable $I^{-1/2}\sum_{i = 1}^{I}\lambda^{T}\{X_{i} - \mu_{i}(\tilde{\varrho})\}$ for a fixed, non-zero, $\lambda \in \R^{K}$.  By independence between strata, the random variables $\lambda^{T}\{X_{i} - \mu_{i}(\tilde{\varrho})\}$ are independent but not necessarily identically distributed. The variance of $I^{-1/2}\sum_{i = 1}^{I}\lambda^{T}\{X_{i} - \mu_{i}(\tilde{\varrho})\}$ is $I^{-1}\sum_{i=1}^I\lambda^{T}\Sigma_i(\tilde{\varrho})\lambda$. By hypothesis $I^{-1}\Sigma(\tilde{\varrho})$ has an positive definite limit $M$ as $I\rightarrow \infty$ so
		\begin{equation}\label{eqn: lyapunov condition part 1}
		\lim\limits_{I \rightarrow \infty}\frac{1}{\left(I^{-1}\sum_{i = 1}^{I}\lambda^{T}\Sigma_{i}(\tilde{\varrho})\lambda \right)^{\frac{2 + \delta}{2}}} = \frac{1}{\left(\lambda^{T}M\lambda \right)^{\frac{2 + \delta}{2}}} > 0.
		\end{equation}
		
		Furthermore, \eqref{eqn: normal distrib criterion} and the $c_{r}$-inequality imply that
		\begin{equation}\label{eqn: lyapunov condition part 2}
		\lim\limits_{I \rightarrow \infty}I^{-\frac{2  +\delta}{2}} \sum_{i = 1}^{I}\Expectation{\mid \sum_{j = 1}^{i}q_{ijk}Z_{ij} - \sum_{j = 1}^{i}q_{ijk}\varrho_{ij} \mid ^{2 + \delta} } = 0.
		\end{equation}
		
		Combining \eqref{eqn: lyapunov condition part 1} and \eqref{eqn: lyapunov condition part 2} gives that
		\begin{equation*}
		\lim\limits_{I \rightarrow \infty}\frac{1}{\left(\sum_{i = 1}^{I}\lambda^{T}\Sigma_{i}(\tilde{\varrho})\lambda \right)^{\frac{2 + \delta}{2}}} \sum_{i = 1}^{I}\Expectation{\mid \sum_{j = 1}^{i}q_{ijk}Z_{ij} - \sum_{j = 1}^{i}q_{ijk}\varrho_{ij} \mid ^{2 + \delta} } = 0.
		\end{equation*}
		
		The Lyapunov central limit theorem then implies that
		\begin{equation*}
		\frac{\sum_{i = 1}^{I}\lambda^{T}\{X_{i} - \mu_{i}(\tilde{\varrho})\}}{\left\{\sum_{i = 1}^{I}\lambda^{T}\Sigma_{i}(\tilde{\varrho})\lambda \right\}^{1/2}} 
		\end{equation*}
		converges in distribution to the standard univariate normal.  Hence, the Cram\'{e}r-Wold device establishes that $\Sigma(\vect{\tilde{\varrho}})^{-1/2}\left\{T - \mu(\vect{\tilde{\varrho}})\right\}$ converges in distribution to a $K$-variate vector of independent standard normals. 
	\end{proof}
	
	The sufficient criterion given in the main text, that $I^{-1}\sum_{i = 1}^{I}\sum_{j = 1}^{n_{i}}q_{ijk}^{4}$ is uniformly bounded for all $I$ and all $k = 1, \ldots, K$, satisfies the conditions of Lemma~\ref{lem: normal distrib} with $\delta = 2$ since $Z_{ij}$ is binary and $0 \leq \tilde{\varrho}_{ij} \leq 1$ for all $i$ and $j$. 
	
	For many statistics, such as an $m$-statistic using Huber's $\psi$ function, $q_{ijk}$ are bounded for all $i, j$ and $k$. In these cases, asymptotic normality would hold if the stratum sizes $n_{i}$ were bounded, for instance. When the underlying $q_{ijk}$ varies as a function of $I$ as with various rank tests, the proof given above is insufficient.  In such cases, a triangular array version of the central limit theorem must be applied and the sufficient conditions adapted accordingly to guarantee asymptotic normality as $I \rightarrow \infty$.

	\begin{proof}[Proof of Theorem~\ref{thm: asymptotic chi bar sqared}]
		Consider the random variable
		\begin{equation}\label{eqn: dual solution to chibarsq2}
		D_{\Lambda_+}^2 = h^{T}\Sigma(\vect{\tilde{\varrho}})h - \inf\limits_{\lambda \in \Lambda_+}(h - \lambda)^{T}\Sigma(\vect{\tilde{\varrho}})(h - \lambda),
		\end{equation}
		where $h = \Sigma(\vect{\tilde{\varrho}})^{-1}\{T - \mu(\vect{\tilde{\varrho}})\}$.  Assume no degeneracy between the test statistics, such that the covariance matrix $\Sigma(\vect{\tilde{\varrho}})$ is positive definite for all $I$. For $\Sigma(\vect{\tilde{\varrho}})$ positive definite, the program
		\begin{align}
		\inf\limits_{\lambda \in \Lambda_+}(h - \lambda)^{T}\Sigma(\vect{\tilde{\varrho}})(h - \lambda) \label{eq:quadratic}
		\end{align} 
		is convex.  Since the feasible region of \eqref{eq:quadratic} is $\Lambda_{+}$, the relative interior of the feasible region is non-empty \citep[\S 2.1.3]{BV04} and Slater's condition holds \citep[\S 5.2.3]{BV04}.  Consequently, there is no duality gap and the Karush-Kuhn-Tucker conditions are both necessary and sufficient for optimality \citep[\S 5.5.3]{BV04}.  As the objective function of \eqref{eq:quadratic} is a quadratic form, it is a smooth function of the arguments $h$, $\lambda$, and $\Sigma(\tilde{\varrho})$.  Thus, the Karush-Kuhn-Tucker conditions stipulate that an optimal $\lambda$ is the root of continuous functions of $h$ and $\Sigma(\tilde{\varrho})$.  Since the solutions to the Karush-Kuhn-Tucker conditions are continuous functions of $h$ and $\Sigma(\tilde{\varrho})$, the optima of \eqref{eq:quadratic} are continuous functions of $h$ and $\Sigma(\tilde{\varrho})$.


		\citet{S03} uses that strong duality holds for (\ref{eqn: dual solution to chibarsq2}) to give rise to the identity
		\begin{equation}\label{eqn: shapiro duality trick}
		D^2_{\Lambda_+} = \sup\limits_{\lambda \in \Lambda_+} \frac{\left[\lambda^{T}\{T - \mu(\vect{\tilde{\varrho}})\}\right]^{2} }{\lambda^{T}\Sigma(\vect{\tilde{\varrho}})\lambda}.
		\end{equation} From this, it is seen by the definition of $A^2_{\Lambda_+}$ in (\ref{eq:A}) of the main text that $D^2_{\Lambda_+} = A^2_{\Lambda_+}$.

		\citet{S03} shows that if $Y$ has a multivariate normal distribution with mean vector $\theta$ and known non-singular covariance matrix $V$ then 
		\begin{equation}\label{eqn: Shapiro result}
		\sup\limits_{\lambda \in \Lambda_+} \frac{\left\{\lambda^{T}(Y - \theta)\right\}^{2}}{\lambda^{T}V\lambda} \sim \ChiBarSq(\V^{-1}, \Lambda_{+}).
		\end{equation}
		Since $I^{-1}\Sigma(\tilde{\varrho}) \rightarrow M$ as $I \rightarrow \infty$, it follows that $I^{1/2}\Sigma(\tilde{\varrho})^{-1/2} \rightarrow M^{-1/2}$.  By Lemma~\ref{lem: normal distrib} the random vector $\Sigma(\vect{\tilde{\varrho}})^{-1/2}\left\{T - \mu(\vect{\tilde{\varrho}})\right\}$ converges in distribution to a $K$-variate vector of independent standard normals. By Slutsky's Lemma $I^{1/2}h$ converges in distribution to the mean-zero multivariate normal distribution with covariance $M^{-1}$. By continuity of the function taking $h$ to the optima of \eqref{eq:quadratic} along with (\ref{eqn: shapiro duality trick}), the mapping 
		\begin{equation*}
		\Sigma(\vect{\tilde{\varrho}})^{-1/2}\left\{T - \mu(\vect{\tilde{\varrho}})\right\} \mapsto \sup\limits_{\lambda \in \Lambda_+} \frac{\left\{\lambda^{T}(T - \mu(\vect{\tilde{\varrho}}))\right\}^{2} }{\lambda^{T}\Sigma(\vect{\tilde{\varrho}})\lambda}
		\end{equation*}
		is continuous.  
		Exploiting Slutsky's Lemma, the Continuous Mapping Theorem, and \eqref{eqn: Shapiro result} yields that $A^2_{\Lambda_+}$ converges in distribution to a $\ChiBarSq(M^{-1}, \Lambda_{+})$ random variable as desired.   
	\end{proof}

	\subsection{Theorem~\ref{thm: nested lambda design sensitivity}}
	\begin{theorem}
		Suppose $\Lambda_{1} \subseteq \Lambda_{2}$. Under mild conditions, the design sensitivity of \eqref{eq:twoplayer} using $\Lambda = \Lambda_{1}$ is less than or equal to the design sensitivity of \eqref{eq:twoplayer} using $\Lambda = \Lambda_{2}$.
	\end{theorem}
	
	\begin{proof}
		Define $\designSens_{\Lambda}$ as the design sensitivity of the test using $a_{\Gamma, \Lambda}^{*}$ as a test statistic.  To avoid triviality, suppose that the design sensitivities $\designSens_{\Lambda_{1}}$ and $\designSens_{\Lambda_{2}}$ both exist; see \citet{ros04} and \citet{R13} for mild conditions for existence of the design sensitivity.  Let $A_{\Gamma, \Lambda_{i}}$ be the random variable giving rise to the observation ${a}^*_{\Gamma, \Lambda_i}$ in \eqref{eq:mod} for $i=1,2$, that is 
		\begin{equation*}
		A_{\Gamma, \Lambda_{i}}^{*} = \underset{\varrho\in\cP_\Gamma}{\min}\;\; \underset{\lambda \in \Lambda_{i}}{\sup}\;\; \frac{\lambda^T\{T-\mu(\varrho)\}}{\{\lambda^T\Sigma(\varrho)\lambda\}^{1/2}}.
		\end{equation*}
		Since $\Lambda_{1} \subseteq \Lambda_{2}$, for any $\varrho$ we have
		\begin{equation*}
		\underset{\lambda \in \Lambda_{1}}{\sup}\;\; \frac{\lambda^T\{T-\mu(\varrho)\}}{\{\lambda^T\Sigma(\varrho)\lambda\}^{1/2}} \leq \underset{\lambda \in \Lambda_{2}}{\sup}\;\; \frac{\lambda^T\{T-\mu(\varrho)\}}{\{\lambda^T\Sigma(\varrho)\lambda\}^{1/2}},
		\end{equation*}
		such that $A_{\Gamma, \Lambda_{1}}^{*} \leq  A_{\Gamma, \Lambda_{2}}^{*}$.  Consider any $\Gamma < \designSens_{\Lambda_{1}}$. By the definition of design sensitivity, for a sensitivity analysis conducted at $\Gamma$ we have that $\P(A_{\Gamma,\Lambda_1} \geq k\mid \cZ)$ tends to one as $I \rightarrow \infty$ for any scalar $k$.  Since $A_{\Gamma, \Lambda_{1}}^{*} \leq  A_{\Gamma, \Lambda_{2}}^{*}$, for any $\Gamma < \designSens_{\Lambda_{1}}$ the power of the test based upon $A_{\Gamma, \Lambda_{2}}^{*}$, $\P(A_{\Gamma,\Lambda_2} \geq k\mid \cZ)$, also tends to one as $I \rightarrow \infty$ for any $k$.  Thus, $\designSens_{\Lambda_{2}} \geq \designSens_{\Lambda_{1}}$ as desired.
		
	\end{proof}
	\section{Additional Simulations}
	\subsection{The general setup of the simulation studies}
	In this section we present additional simulation studies to further illustrate the results presented in the manuscript. All of the simulation studies are conducted with some number $I$ pairs, and some number $K$ outcome variables, equicorrelated with correlation controlled by a parameter $\rho$. For each outcome variable, the employed test statistic is $T_k = \sum_{i=1}^I\sign(Y_{ik})\min(|Y_{ik}|/s_k, 2.5)$, where $s_k$ is the median of $|Y_{ik}|$ $(i=1,..,I)$. This amounts to a choice of a $m$-statistic with Huber's $\psi$-function, as described in \citet{ros07}. 
	\subsection{Rejecting the global null with $I=1000$ pairs}
	In \S~\ref{sec: Power Simulations} of the main text the $\ChiBarSq$-test was compared to the equal-weight test and the test of \citet{FS16} with $I=300$ matched pairs.  To highlight the large-sample properties of the test, we include Figure~\ref{fig: Normal Huber power comparisons: Large Sample}. As $I$ increases, the power curves converge pointwise to step functions, evaluating to 1 if $\Gamma$ is below the design sensitivity and zero otherwise \citep{R13}.  This indicates that the gap between the equal-weight test and the  $\ChiBarSq$-test observed in the first row of Figure~\ref{fig: Normal Huber power comparisons} and Figure~\ref{fig: Normal Huber power comparisons: Large Sample} will shrink as $I$ increases, and will disappear in the limit.  This trend can be appraised visually by comparing the disparity observed in the first row of Figure~\ref{fig: Normal Huber power comparisons} in the manuscript where $I = 300$ to the first row of Figure~\ref{fig: Normal Huber power comparisons: Large Sample} where $I = 1000$.  As a consequence Theorem~\ref{thm: nested lambda design sensitivity} and of the pointwise convergence of the power curve to the indicator function of the event $\Gamma$ less than the design sensitivity, the power curve of the $\ChiBarSq$-test will converge to that of the most powerful test at any fixed $\Gamma$ among all coherent tests.
	\begin{figure}
		\centering
		\includegraphics[scale=.65]{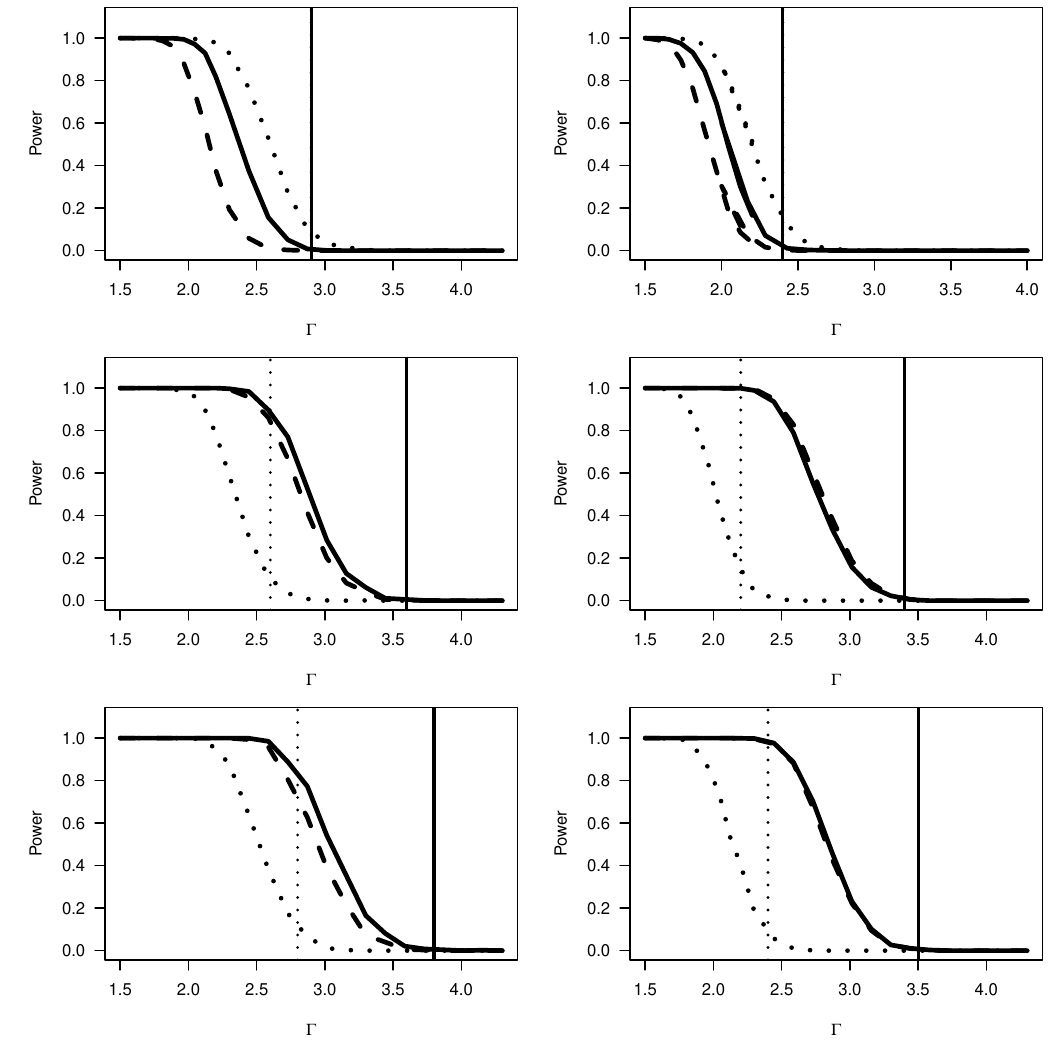}
		\caption{Power comparisons between the method of \citet{FS16} (dashed), the method of this paper (solid), and the equal-weight test (dotted) as $\Gamma$ increases with $I = 1000$.  The first row has $\tau_{1} = \tau_{2} = \tau_{3} = 0.25$.  The second row has $\tau_{1} = \tau_{2} =0.1$ and $\tau_{3} = 0.5$.  The third row has $\tau_{1} = 0.05$,  $\tau_{2} =0.2$, and $\tau_{3} = 0.5$.  Figures on the left have $\rho = 0$ while on the right $\rho = 0.2$. For each fixed set of parameters, power simulations were performed on 1000 simulated data sets.  The design sensitivity of the equal-weight test is the dotted vertical line and the design sensitivity of the $\ChiBarSq$-test is the solid vertical line.  In the first row, these two design sensitivities are the same and are shown by the single solid vertical line. }
		\label{fig: Normal Huber power comparisons: Large Sample}
	\end{figure}

	\subsection{Rejecting individual nulls through closed testing}

	An experimenter may want to test not only the global null hypothesis $H_{0}$ of \eqref{eq:global} but also the $K$ individual null hypotheses $H_{1}, \ldots, H_{K}$.  To achieve this at level $\alpha$, she may use a closed-testing framework \citep{mar76}.  Then, in order to test $H_{i}$ at level $\alpha$, she performs $\alpha$-level tests all hypotheses of the form $H_{i} \wedge \left(\bigwedge_{k \in S_{i}} \right)$ with $S_{i}$ the set of all possible subsets of the numbers $1, \ldots, K$ excluding $i$; she then rejects $H_{i}$ if all of these tests rejected.  Another standard method to test both the global null and each individual null would be to conduct a Bonferroni-corrected test of the global null and then use the results of the corrected individual tests to reject each $H_{k}$.  Figure~\ref{fig: Closed test vs Bonferroni vs Uncorrected} examines the performance of these two methods against the test of only $H_{1}$ when $\tau_{1} = 0.5$, $\tau_{2} = 0.2$, $\tau_{3} = 0.05$ and equicorrelation between the paired differences at at $\rho = 0.2$.  The comparison to the test of only $H_{1}$ is an unfair comparison in that testing only $H_{1}$ at level-$\alpha$ does not control the family-wise error rate at $\alpha$ when examining all $k = 1, \ldots, K$.  However, the test of $H_{1}$ alone at level-$\alpha$ achieves the highest power possible for any testing procedure that tests $H_{k}$ as it does not employ any corrections to control the family-wise error rate.  Thus, comparison to the test of $H_{1}$ alone at level-$\alpha$ serves as a comparison to an idealized benchmark, the absolute limit of statistical power that one may achieve when testing $H_{1}$ using a particular test statistic.
	
	\begin{figure}
		\centering
		\includegraphics[scale=.6]{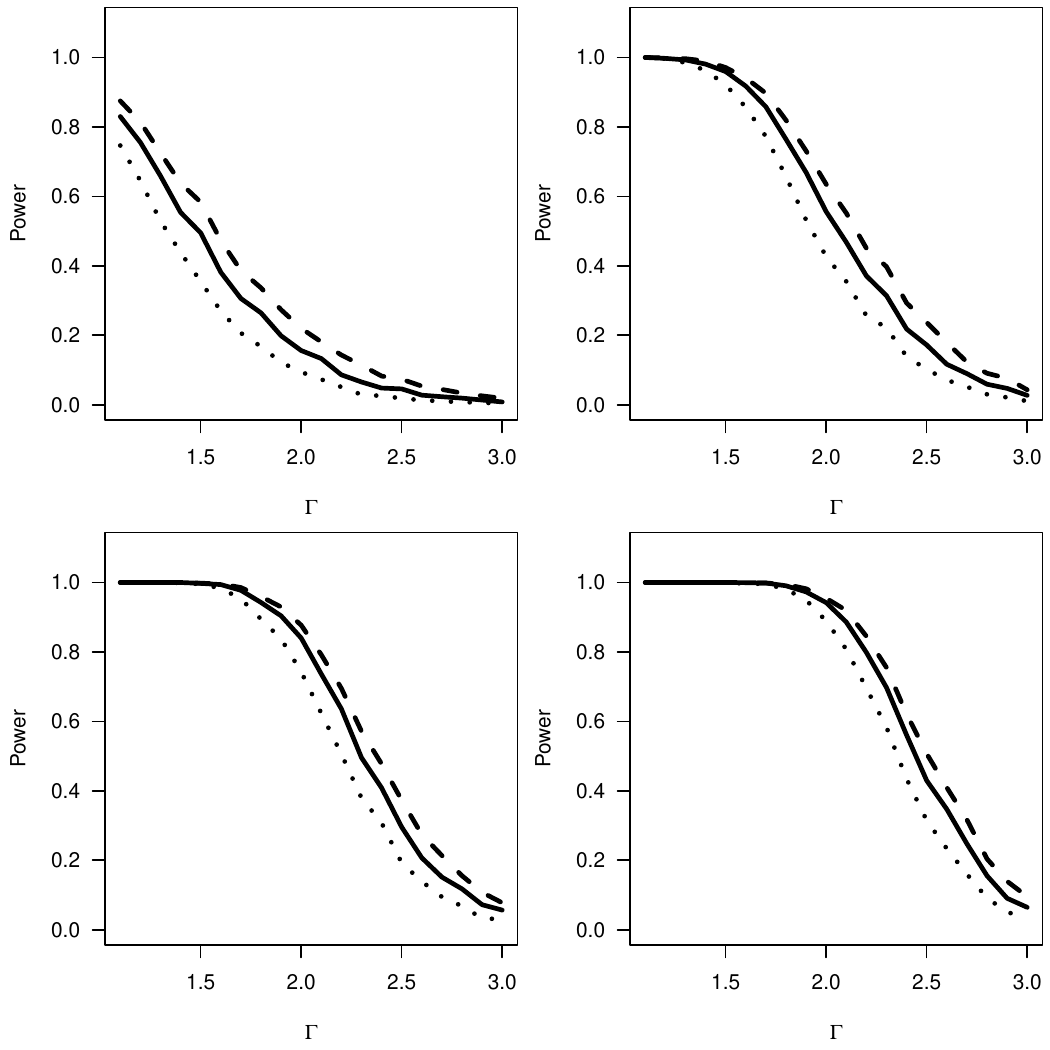}
		\caption{Power comparisons between embedding the $\ChiBarSq$-test into a closed testing framework (solid), performing a Bonferroni-corrected test (dotted), and performing an uncorrected test (dashed) as $I$ increases.  All simulations performed with $\tau_{1} = 0.5$, $\tau_{2} = 0.2$, $\tau_{3} = 0.05$ with equicorrelation at $\rho = 0.2$ testing $H_{1}$.  All data is with normal noise and tested with Huber's $\psi$-function as the underlying statistic.  Additional parameters listed clockwise from the top-left: $I = 50$, $I = 150$, $I = 250$, and $I = 350$.  For each sample size, power simulations were performed on 2000 simulated data sets.}
		\label{fig: Closed test vs Bonferroni vs Uncorrected}
	\end{figure}

	At all values of $I$ examined, the closed test outperforms the Bonferroni-corrected test.  Furthermore, as $I$ increases, the closed test approaches the same power as the individual test without correction.  Thus, for sufficiently large studies, empirical results suggest that a closed testing framework allows the experimenter to test both the global null and the individual null at level-$\alpha$ with minimal loss of power from multiple comparisons relative to testing only the individual null.
	
	\subsection{Type I error control in small samples using the asymptotic reference distribution}
	\begin{table}
		\centering
		\begin{tabular}{ll ccc}  $\tau_{2}$ & $\rho$ & \citet{FS16} & $\Lambda = \Lambda_{+}$ & $\Lambda = \R^{K} \setminus \{0_K\}$\\
			-0.5 & 0 & 0 & 0 & 0.694 \\  -0.25 & 0 & 0 & 0 & 0.454 \\  0 & 0 & 0.026 & 0.018 & 0.46 \\ -0.5 & 0.5 & 0 & 0 & 0.546 \\  -0.25 & 0.5 & 0 & 0 & 0.424 \\  0 & 0.5 & 0.018 & 0.016 & 0.496 \\ 
		\end{tabular}
		\caption{Type I error rates for the method of \citet{FS16}, the $\ChiBarSq$-test of this paper, and the test taking $\Lambda = \R^{K} \setminus \{0_K\}$ using $\alpha = 0.05$.  All tests performed with $I = 20$ matched pairs, $\tau_{1} = -0.5$, and $\Gamma = 1$.  For each set of parameters the power was estimated based upon 500 simulations.}
		\label{tab: type 1 error rate}
	\end{table}
	
	In this simulation, we assess the Type I error rate with $I=20$ matched pairs at $\Gamma=1$. In each simulation, the global null of non-positive treatment effects is true. Table~\ref{tab: type 1 error rate} details simulated Type I error rates for the method of \citet{FS16}, the $\ChiBarSq$-test of this paper, and the unconstrained test taking $\Lambda = \R^{K} \setminus \{0_K\}$ for $K=2$ outcome variables with a range of different parameter values.

	Both the method of \citet{FS16} and the $\ChiBarSq$-test control the Type I error rate at $\alpha$ even when in the finite sample regime while using the asymptotic reference distribution. As alluded to in \S \ref{sec: Nonnegative Orthant} the test taking $\Lambda = \R^{K} \setminus \{0_K\}$ fails to control the Type I error rate at $\alpha$ since allowing $\lambda$ to have an unconstrained sign in each coordinate removes the ability to discriminate positive treatment effects from negative treatment effects. This further motivates the restriction to the set of coherent combinations $\Lambda_+.$
	
	\subsection{Non-normal and larger $K$ simulations}
	We include several additional simulations, using a larger number of outcomes and experimenting with heavy-tailed noise in the data generating distribution.  In order to demonstrate the method's properties on studies with more outcomes, we conducted tests with $K = 4$.  Choosing $K = 4$ still allows for interpretable regimes of treatment effect relative magnitudes while expanding from the trivariate case.  We conducted tests under normality as in Section 5 of the manuscript as well as under non-normal conditions.
	
	To generate the normal $4$-variate data we followed the same procedure as outlined in Section 5.2 of the manuscript, but using four $\tau$'s and four $\varepsilon$'s.  To conduct the non-normal tests, we elected to experiment with heavy-tailed noise.  This was implemented via substituting $t_{5}$-distributed noise $(\varepsilon_{1}, \ldots, \varepsilon_{4})$ in place of normal noise in the procedure of Section 5.2.  This process mirrored the $t$-distributed simulation construction of \citet{ros16}.  We conducted tests for $\tau = (.1, .1, .1, .5), (.1, .1, .5 , .5),$ and $(.1, .25, .25, .5)$.  These treatment effect relative magnitudes were selected to highlight the strength of the $\ChiBarSq$-test when:
	\begin{itemize}
		\item One treatment effect is much larger than the others but none are of negligible magnitude (this is the case of $\tau = (.1, .1, .1, .5)$).
		\item There are several highly impacted outcomes, but there remain several outcomes for which treatment effect is small (this is the case of $\tau = (.1, .1, .5, .5)$).  This regime does not exist in the trivariate case.
		\item The treatment effects are spread across multiple magnitude scales; selecting all to be equally weighted is apt to perform poorly, but selecting only the largest is unlikely to perform as well as optimizing for weighting in accordance with their magnitudes (this is the case of $\tau = (.1, .25, .25, .5)$).
	\end{itemize}
	For all of the test considered, $I = 300$, $n_{i} = 2$ for all $i$, and $\rho = 0.2$. Figure~\ref{fig: Normal and T power comparisons} presents the results of these simulations (cf. Figure 1 of the manuscript).
	
	\begin{figure}[h]
		\centering
		\includegraphics[scale=.65]{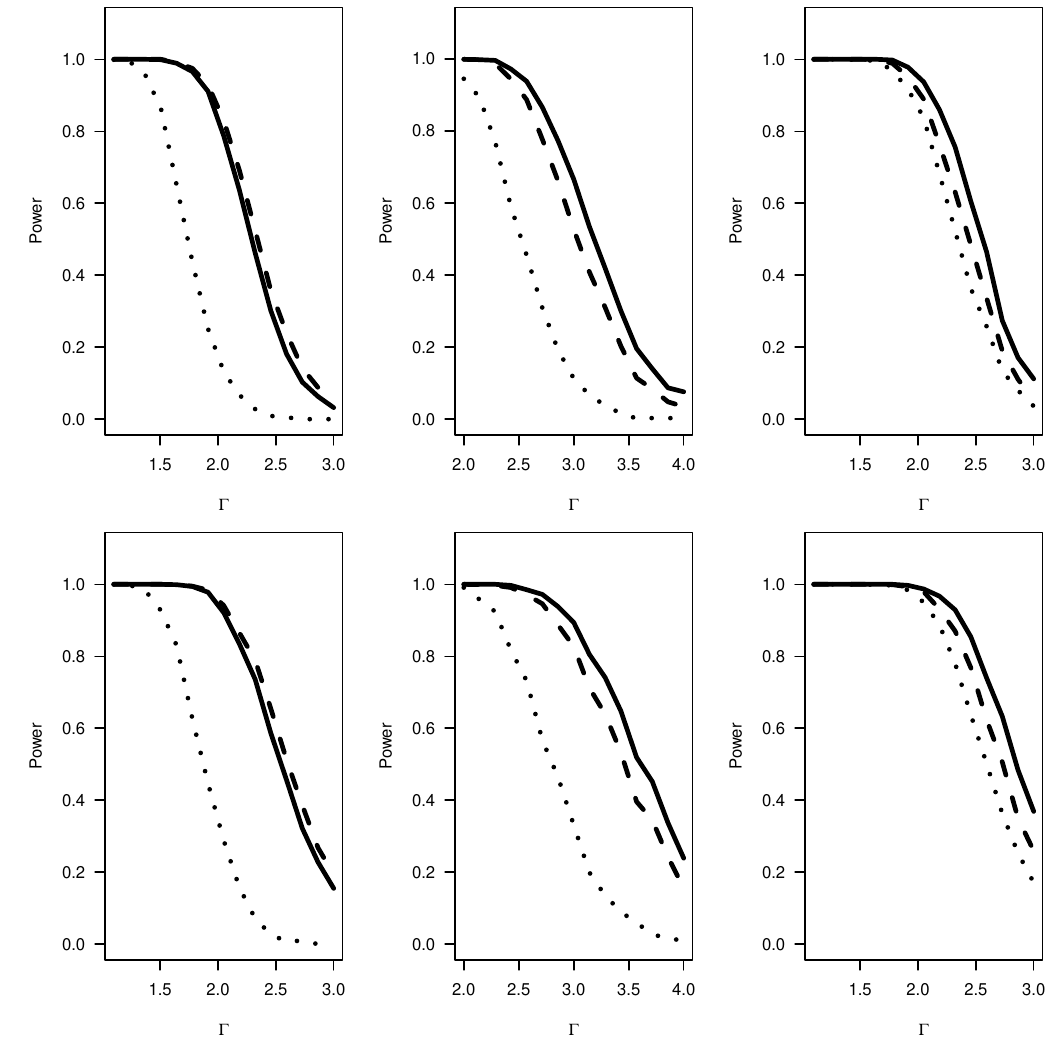}
		\caption{Power comparisons between the method of \citet{FS16} (dashed), the $\ChiBarSq$-test of this paper (solid), and the equal-weight test (dotted) as $\Gamma$ increases with $I = 300$. The first column has $\tau = (.1, .1, .1, .5)$; the second column has $\tau = (.1, .1, .5, .5)$; and the third column has $\tau = (.1, .25, .25, .5)$.  The top row is generated under the Gaussian data-generating process and the bottom row is generated under the $t_{5}$ data-generating process.}
		\label{fig: Normal and T power comparisons}
	\end{figure}
	
	Despite the heavy-tails, the relative performance of the $\ChiBarSq$-method stays the same as under the Gaussian data-generating process.  Moreover, the performance of the $\ChiBarSq$-test relative to the equal-weight test and the basis-vector test accords well with the intuition developed in Section 5 of the manuscript.
	\begin{itemize}
		\item In the first column, the equal-weight test is apt to under-perform due to the strong disparity in treatment effects across the four outcomes.  Since there is one ``stand-out" effect the lower critical value of the basis-vector test accounts for the slight increase performance edge over the $\ChiBarSq$-test.
		\item In the second column, the equal-weight test fares poorly for the same reasons as before.  Since there are several outcomes that are strongly impacted, the $\ChiBarSq$-test outperforms the basis-vector test which weights only one outcome.
		\item In the third column, the spread of treatment effect magnitudes across different regimes again accounts for the strong performance of the $\ChiBarSq$-statistic over the other two, as it flexibly weights each outcome in accordance with the degree of treatment effect.
	\end{itemize}

	
	\section{Algorithmic Details for Conducting the Sensitivity Analysis}\label{sec: Optimization Implementation Details}
	The optimization problem in \eqref{eq:mod} is solved via a projected subgradient descent algorithm.  \citet{Shor85} contains a detailed introduction to subgradient methods.  The algorithm begins with some initial feasible $\vect{\varrho}_{(0)}$, solves for an optimal $\lambda$ under the fixed $\vect{\varrho}_{(0)}$, computes a subgradient of the objective at the optimal $\lambda$, and uses the subgradient to project onto the feasible region thereby locating a $\vect{\varrho}_{(1)}$.  The procedure iterates until convergence criteria are satisfied.
	
	Formally, given a feasible $\varrho_{(n)}$ we compute
	\begin{equation}\label{eqn: lambda update formula}
	\lambda_{\varrho_{(n)}}^{*} = \sup\limits_{\lambda \in \Lambda_+} \frac{\left[\max\{0, \lambda^{T}(T - \mu(\varrho_{(n)}))\}\right]^{2}}{\lambda^{T}\Sigma(\varrho_{(n)})\lambda}.
	\end{equation}
	To compute \eqref{eqn: lambda update formula}, results from \citet{S03} are leveraged to allow efficient computation of
	\begin{equation*}
	\sup\limits_{\lambda \in \Lambda_+} \frac{\lambda^{T}\{T - \mu(\varrho_{(n)})\}}{\{\lambda^{T}\Sigma(\varrho_{(n)}\lambda\}^{1/2}},
	\end{equation*}
	by solving a single quadratic program.  In the event that 
	\begin{equation*}
	\sup\limits_{\lambda \in \Lambda_+} \frac{\lambda^{T}\{T - \mu(\varrho_{(n)})\}}{\{\lambda^{T}\Sigma(\varrho_{(n)})\lambda\}^{1/2}} > 0
	\end{equation*}
	$\lambda_{\varrho_{(n)}}^{*}$ is set to the optimizing choice of $\lambda$.  However, when 
	\begin{equation*}
	\sup\limits_{\lambda \in \Lambda_+} \frac{\lambda^{T}\{T - \mu(\varrho_{(n)})\}}{\{\lambda^{T}\Sigma(\varrho_{(n)})\lambda\}^{1/2}} \leq 0
	\end{equation*}
	there exists a feasible $\vect{\varrho}$ such that the test fails to reject the sharp null, and thus no further iterations of the subgradient method are needed.
	
	By \citet{HUL13}, if $f(x) = \sup_{j \in J}f_{j}(x)$ where each $f_{j}(x)$ is a convex function, $f(x) = f_{j^{*}}(x)$, and $g \in \partial f_{j^{*}}(x)$, then $g \in \partial f(x)$.  In less technical terms, to compute a subgradient of a function which is the point-wise supremum of a many convex functions, one first finds a function $f_{j^{*}}(\cdot)$ which achieves the maximum value at the desired point $x$ and then one computes a subgradient of this function.  As such, at the optimal value $\lambda^{*}$ one computes that the subgradient of the objective function with respect to the variables $\varrho_{i} = (\varrho_{i1}, \ldots, \varrho_{in_{i}})$ is
	
	\begin{equation*}\label{eqn: subgradient equation}
	g = \frac{h_{1}(\varrho)\partial_{\varrho_{i}}h_{2}(\varrho) - h_{2}(\varrho)\partial_{\varrho_{i}}h_{1}(\varrho)}{h_{2}(\varrho)^{2}},
	\end{equation*}
	where 
	\begin{align*}
	h_{1}(\varrho) &= \left(\lambda^{*T}\left(T - \mu(\varrho)\right)\right)^{2},\\
	h_{2}(\varrho) &= \lambda^{*T}\Sigma(\varrho)\lambda^{*},\\
	\partial_{\varrho_{i}}h_{1}(\varrho) &= -2(Q_{i}^{T}\lambda^{*})\lambda^{*T}(T - \mu(\varrho)),\\
	\partial_{\varrho_{i}}h_{2}(\varrho) &= (Q_{i}^{T}\lambda^{*}) \circ (Q_{i}^{T}\lambda^{*}) - 2(Q_{i}^{T}\lambda^{*})(Q_{i}^{T}\lambda^{*})^{T}\varrho_{i},
	\end{align*}
	where $Q_{i}$ is the $K$-by-$n_{i}$ matrix where the $(k,j)$th entry is $q_{ijk}$ and $\circ$ denotes the coordinate-wise product operation.
	
	Armed with the solution to the inner maximization and the form of the subgradient $g$, we can now detail the projected subgradient descent method.
	
	\begin{enumerate}
		\item Initialize a feasible $\rho_{(0)}$, pick $t_{0} > 0$ and $n = 1$
		\item Repeat until convergence:
		\begin{enumerate}
			\item Find $\lambda_{\rho_{(n - 1)}}$ by solving \eqref{eqn: lambda update formula},
			\item Compute the subgradient $g$ from \eqref{eqn: subgradient equation} using $\lambda_{\rho_{(n - 1)}}$,
			\item Define $\varrho_{(n)}$ to be the projection of $\rho_{(n - 1)} - t_{n - 1}g$ onto the feasible region,
			\item Update the parameters: $t_{n} = t_{0} / \sqrt{n}$ and $n = n + 1$.
		\end{enumerate}
	\end{enumerate}
	
	Since the objective function is convex and the feasible set is also convex, any local optimum is a global optimum as well.  In practical execution on both synthetic and real data sets convergence has been observed after few iterations.
	
	\section{The $\ChiBarSq$ Distribution}
	\subsection{Finding a better critical value}
	While the subgradient method solves \eqref{eq:mod} and Theorem~\ref{thm: asymptotic chi bar sqared} gives that the asymptotic distribution of $A_{\Lambda_{+}}^{2}$ is $\ChiBarSq$, the weights of the limiting distribution are still unknown.  Comparing the value of \eqref{eq:mod} against the $1 - \alpha$ quantile arising from the bound 
	\begin{align} \P\{\ChiBarSq(V, \Lambda_+) \geq c\} \leq 0.5 \{\Prob{\chi_{K - 1}^{2} \geq c} + \Prob{\chi_{K}^{2} \geq c}\}\label{eq:naive}\end{align} would control the Type I error. While improving over a critical value based on a $\chi^2_K$ distribution, the bounds through (\ref{eq:naive}) are still unduly conservative. We now describe an algorithm which exploits the particular structure of the sensitivity analysis problem to dramatically improve the critical value.
	
	By directly computing upper and lower bounds on the correlation between $T_{k}$ and $T_{\ell}$ for each $k, \ell = 1, \ldots, K$ one can compute coordinate-wise upper and lower bounds on on the overall correlation matrix $\diag\{\Sigma(\vect{\varrho})\}^{-1/2}\Sigma(\vect{\varrho})\diag\{\Sigma(\vect{\varrho})\}^{-1/2}$, where $\diag\{\Sigma(\vect{\varrho})\}$ contains the diagonal elements of $\Sigma(\vect{\varrho})$ on its diagonals but has zeroes on its off-diagonals.  Since the weights of the $\ChiBarSq$ distribution depend on $\Sigma(\vect{\varrho})$ only through its correlation matrix \citep{SS05}, one can directly optimize over bounds on the marginal correlations to find the most conservative $1 - \alpha$ critical value associated to a correlation matrix within the bounds. This optimization can be performed via either numerical approximation of gradients or by directly computing gradients of the $p$-value function with respect to the correlations.  Such gradients are accessible due to Plackett's identity \citep{P54} and can be calculated with assistance of functions in the \texttt{mvtnorm} package within \texttt{R} for evaluating orthant probabilities and the density of the multivariate normal. Optimizing over the space of correlation matrices yields significant improvement over the critical value drawn from previous bound.  Figure~\ref{fig: Critical Value Gap} highlights the differences between using the $1-\alpha$ quantile from a $\chi^{2}$ distribution, the $1-\alpha$ quantile from the naive bound based upon (\ref{eq:naive}), and using the optimal $1-\alpha$ quantile with $K=3$ outcome variables. By using the most conservative $1 - \alpha$ quantile within the upper and lower bounds on the correlation matrix the Type I error rate is asymptotically controlled at $\alpha$.  
	
	\begin{figure}[h]
		\centering
		\begin{subfigure}[t]{0.45\textwidth}
			\centering
			\includegraphics[width = \linewidth]{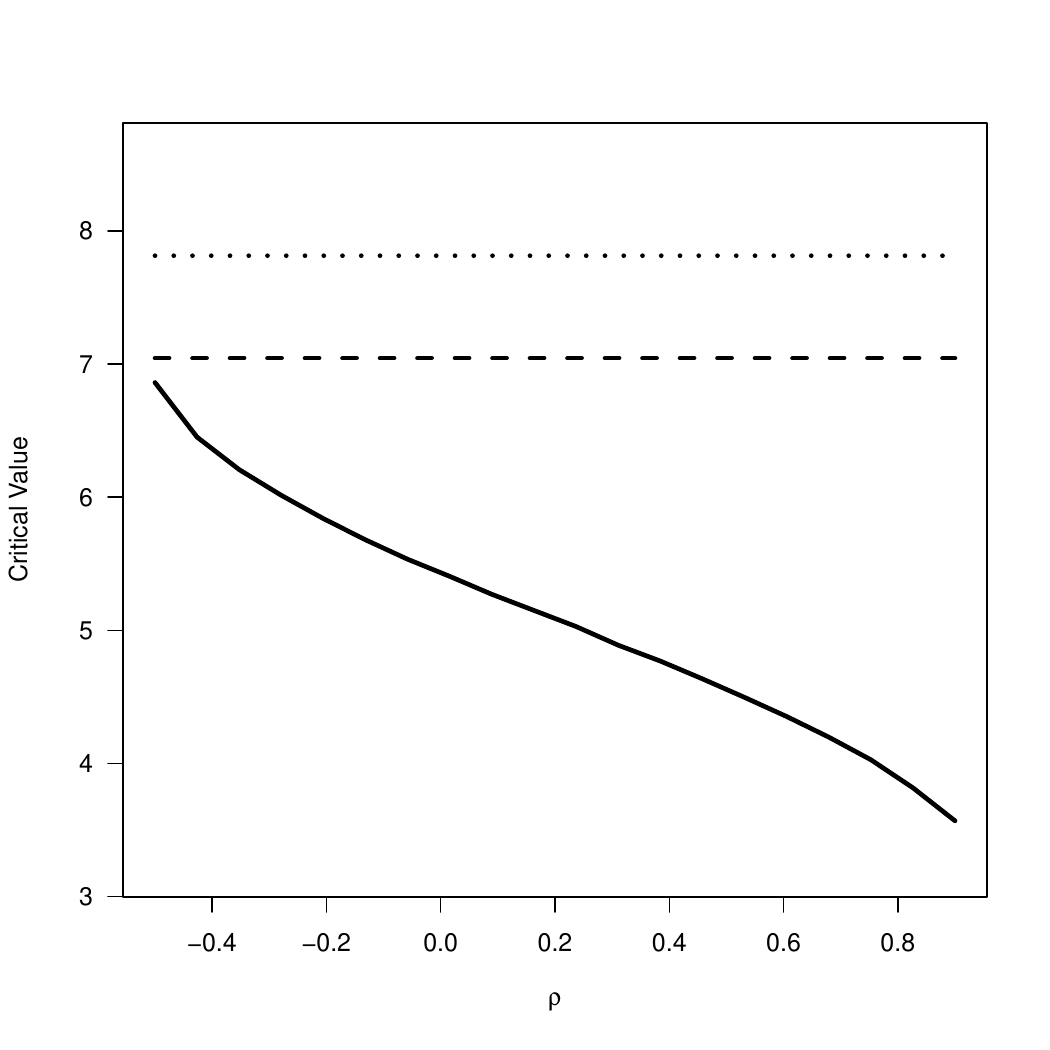}
		\end{subfigure}
		\hfill
		\begin{subfigure}[t]{0.45\textwidth}
			\centering
			\includegraphics[width = \linewidth]{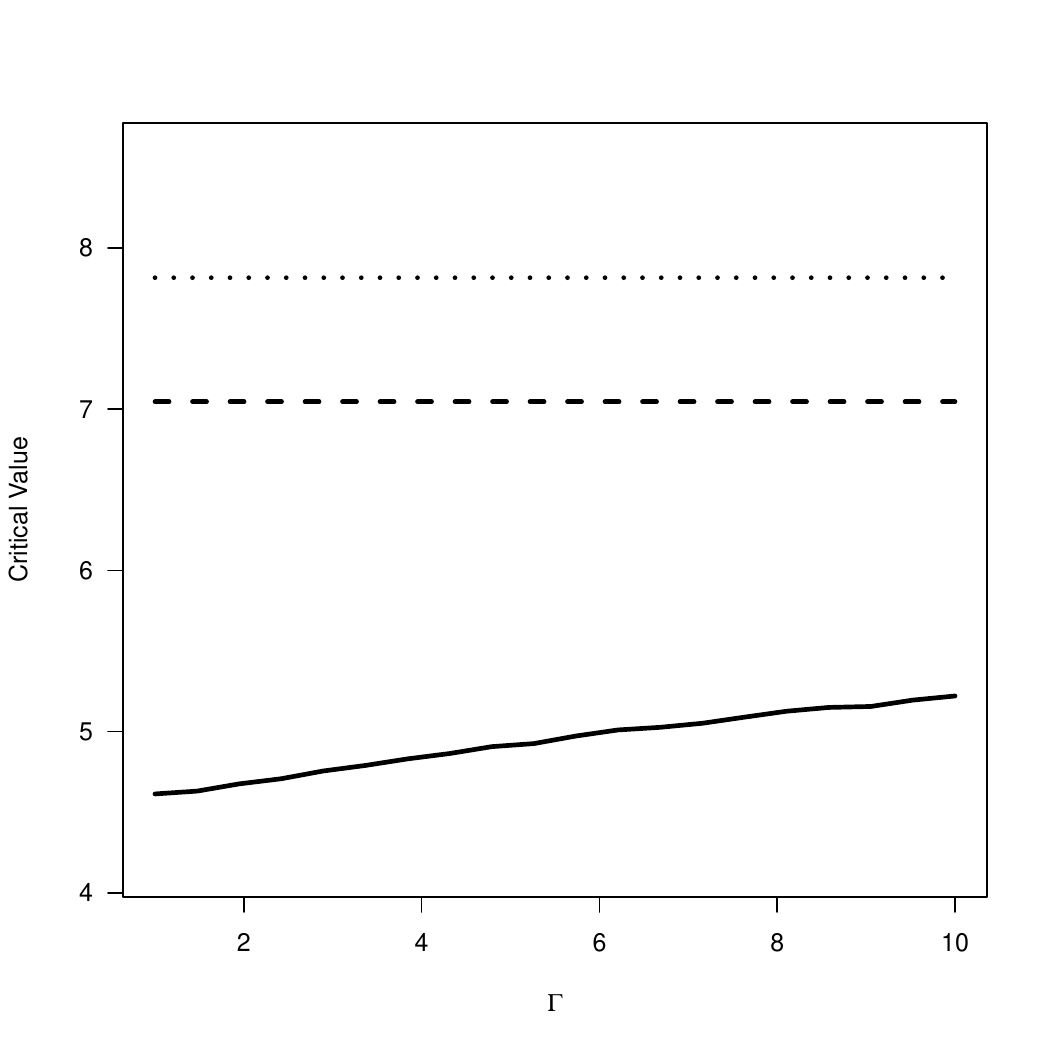}
		\end{subfigure}
		\caption{$1-\alpha$ quantiles for $\alpha = 0.05$ generated for the trivariate scenario $I = 300$, $\tau_{1} = \tau_{2} = \tau_{3} = 0$.  On the left, $\Gamma$ is fixed at 1 while $\rho$ varies over $[-0.5, 0.9]$.  On the right, $\rho$ is fixed at $0.5$ while $\Gamma$ ranges from $1$ to $10$.  In both figures the $\chi^{2}_3$ $1-\alpha$ quantile is the dotted line, that of the naive bound derived from (\ref{eq:naive}) is the dashed line, and the $1-\alpha$ quantile coming from optimizing over feasible correlation matrices is the solid line.}
		\label{fig: Critical Value Gap}
	\end{figure}
	
	There is a true, but generally unknown, underlying correlation structure between the test statistics that depends upon the true vector of conditional probabilities $\tilde{\varrho}$.  Thus, the true $1-\alpha$ quantile from the $\ChiBarSq$ distribution with weights based on the true correlation would not change with the value of $\Gamma$ employed in the sensitivity analysis.  As the true unmeasured confounders are unknown, we instead find a conservative critical value based upon the feasible values for $\varrho$ at a given $\Gamma$.  As $\Gamma$ grows so too does the feasible region for the probabilities $\mathcal{P}_\Gamma$; consequently the conservative critical value increases with $\Gamma$ as well.  This explains the trend in the right-hand panel of Figure~\ref{fig: Critical Value Gap}.

	\subsection{The worst-case correlation with bivariate outcomes}
	In the case for $K = 2$, an elementary proof establishes a closed form of the optimizing correlation matrix subject to box constraints.  
	\begin{theorem}
		Suppose that $K = 2$ and $[\ell, u] \subseteq (-1, 1)$.  Over all matrices $M$ in the set
		\begin{equation*}
		S = \left\{ \begin{bmatrix}1 & \rho \\
		\rho & 1 \end{bmatrix} \, : \, \rho \in [\ell, u] \right\}
		\end{equation*}
		the matrix $\begin{bmatrix}1 & \ell \\
		\ell & 1 \end{bmatrix}$ achieves the most conservative (largest possible) $1 - \alpha$ quantile of $\ChiBarSq(M^{-1}, \Lambda_+)$.
	\end{theorem}
	\begin{proof}
		Say that $X \sim \ChiBarSq(M^{-1}, \Lambda_+)$ for $M^{-1} \in S$.  From \citet{SS02} the probability
		\begin{align*}
		\Prob{X \leq c} &= w_{0}(2, M^{-1})\Prob{\chi_{0}^{2} \leq c} +  
		\frac{1}{2}\Prob{\chi_{1}^{2} \leq c} +  w_{2}(2, M^{-1})\Prob{\chi_{2}^{2} \leq c}\\
		&= w_{2}(2, M)\Prob{\chi_{0}^{2} \leq c} +  
		\frac{1}{2}\Prob{\chi_{1}^{2} \leq c} +  w_{0}(2, M)\Prob{\chi_{2}^{2} \leq c}
		\end{align*}
		where each $\chi_{i}^{2}$ is an independent random variable with $\chi_{i}^{2}$ distribution.  The value $w_{2 - i}(2, M)$ is the probability that the projection, under the norm induced by the quadratic form $x^{T}Mx$, of a standard bivariate normal random vector onto the non-negative orthant has exactly $2 - i$ positive components.  By an argument presented in \citet{SS02}, this interpretation of $w_{2 - i}(2, M)$ is equivalent to defining $w_{2 - i}(2, M)$ as the probability that a standard bivariate normal random variable $Z$ falls into $R_{i} = \left\{x \in \R^{2} \given \sum_{k=1}^2 1(b_k > 0) = i\right\}$ where $b = M^{1/2}z$.  Since $w_{2}(2, M) +  w_{0}(2, M) = 1$ and $\Prob{\chi_{0}^{2} \leq c} \geq \Prob{\chi_{2}^{2} \leq c}$ for all scalars $c$ it suffices to maximize $w_{2}(2, M)$.
		
		Taking 
		\begin{equation*}
		M = \begin{bmatrix}1 & \rho \\ \rho & 1 \end{bmatrix}
		\end{equation*}
		gives that maximizing $w_{2}(2, M)$ is equivalent to maximizing the area $R_{2}$ in Figure~\ref{fig: chibarsq bivariate proof}
		\begin{figure}[H]
			\centering
			\begin{tikzpicture}
			\draw[dashed,->] (0,0) -- (2,0) node[anchor = west] {$x_{1}$};
			\draw[dashed,->] (0,0) -- (0,2) node[anchor = south] {$x_{2}$};
			\draw[dashed,->] (0,0) -- (-2,0);
			\draw[dashed,->] (0,0) -- (0,-2);
			
			\coordinate (A) at (1, 2) {};
			\coordinate (B) at (2, 0) {};
			\coordinate (C) at (-2, 1) {};
			\coordinate (D) at (0, -2) {};
			\coordinate (0) at (0, 0) {};
			
			\draw[thin,->] (0,0) -- (1,2);
			\draw[thin,->] (0,0) -- (2,0);
			
			\draw[thin,->] (0,0) -- (-2,1);
			\draw[thin,->] (0,0) -- (0,-2);
			
			\tkzMarkRightAngle[draw=black,size=.2](A,0,C);
			\tkzMarkRightAngle[draw=black,size=.2](B,0,D);
			
			\draw (1.5,1) node[anchor = center] {$R_{2}$};
			
			\draw (A) node[anchor=south west] {$A$};
			\end{tikzpicture}
			\caption{Pictorial representation of the region $R_{2}$.  The upper right boundary of $R_{2}$ is the line given by $A$.}
			\label{fig: chibarsq bivariate proof}
		\end{figure}
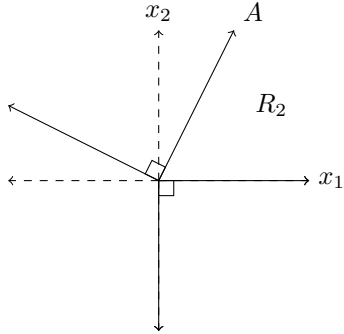
		The slope of the line $A$ is $\rho - \rho^{-1}$ when $\rho \neq 0$ and $A$ is vertical when $\rho = 0$.  Maximizing $R_{2}$ corresponds to taking $\rho$ as small as possible within $[\ell, u]$.  Thus the matrix $\rho = \ell$ achieves the most conservative $1 - \alpha$ critical value of $\ChiBarSq(M^{-1}, \Lambda_+)$.
		
	\end{proof}

	\subsection{A bivariate illustration of the $\ChiBarSq$ distribution}
	Consider a mean-zero bivariate normal with covariance $V$ and consider the distribution of $\ChiBarSq(V, \Lambda_+)$.  By the law of total probability, $\Prob{\bar{\chi}^2(V, \Lambda_+) \leq c} = \sum_{i = 0}^{2}\P\{\bar{\chi}^2(V, \Lambda_+) \leq c \given X \in R_{i}\}\P(X \in R_{i})$ where $R_{0}, R_1$, and $R_{2}$ are disjoint coverings of $\R^{2}$. Let $b=V^{-1/2}x$, and set $R_{i} = \left\{x \in \R^{2} \given \sum_{k=1}^2 1(b_k > 0) = i\right\}$; this is shown in Figure~\ref{fig: chibarsq example regions}.
	\begin{figure}[H]
		\centering
		\begin{subfigure}[t]{0.49\textwidth}
			\centering
			\begin{tikzpicture}
			\draw[dashed,->] (0,0) -- (2,0) node[anchor = west] {$x_{1}$};
			\draw[dashed,->] (0,0) -- (0,2) node[anchor = south] {$x_{2}$};
			\draw[dashed,->] (0,0) -- (-2,0);
			\draw[dashed,->] (0,0) -- (0,-2);
			
			\coordinate (A) at (0, 2) {};
			\coordinate (B) at (2, 0) {};
			\coordinate (C) at (-2, 0) {};
			\coordinate (D) at (0, -2) {};
			\coordinate (0) at (0, 0) {};
			
			\draw[thin,->] (0,0) -- (A);
			\draw[thin,->] (0,0) -- (B);
			
			\draw[thin,->] (0,0) -- (C);
			\draw[thin,->] (0,0) -- (D);
			
			\tkzMarkRightAngle[draw=black,size=.2](A,0,C);
			\tkzMarkRightAngle[draw=black,size=.2](B,0,D);
			
			\draw (1.5,1.5) node[anchor = center] {$R_{2}$};
			\draw (-1.5,1.5) node[anchor = center] {$R_{1}$};
			\draw (-1.5,-1.5) node[anchor = center] {$R_{0}$};
			\draw (1.5,-1.5) node[anchor = center] {$R_{1}$};
			\end{tikzpicture}
		\end{subfigure}
		\hfill
		\begin{subfigure}[t]{0.49\textwidth}
			\centering
			\begin{tikzpicture}
			\draw[dashed,->] (0,0) -- (2,0) node[anchor = west] {$x_{1}$};
			\draw[dashed,->] (0,0) -- (0,2) node[anchor = south] {$x_{2}$};
			\draw[dashed,->] (0,0) -- (-2,0);
			\draw[dashed,->] (0,0) -- (0,-2);
			
			\coordinate (A) at (1, 2) {};
			\coordinate (B) at (2, 1) {};
			\coordinate (C) at (-2, 1) {};
			\coordinate (D) at (1, -2) {};
			\coordinate (0) at (0, 0) {};
			
			\draw[thin,->] (0,0) -- (1,2);
			\draw[thin,->] (0,0) -- (2,1);
			
			\draw[thin,->] (0,0) -- (-2,1);
			\draw[thin,->] (0,0) -- (1,-2);
			
			\tkzMarkRightAngle[draw=black,size=.2](A,0,C);
			\tkzMarkRightAngle[draw=black,size=.2](B,0,D);
			
			\draw (1.5,1.5) node[anchor = center] {$R_{2}$};
			\draw (-.5,1.5) node[anchor = center] {$R_{1}$};
			\draw (-1.5,-1.5) node[anchor = center] {$R_{0}$};
			\draw (1.5,-.5) node[anchor = center] {$R_{1}$};
			\end{tikzpicture}
		\end{subfigure}
		\caption{The regions corresponding to different distributional forms of the likelihood ratio statistic.  In the left image $V = I_{2 \times 2}$; the right image illustrates the general case, in this case the correlation is $-0.8$.}
		\label{fig: chibarsq example regions}
	\end{figure}
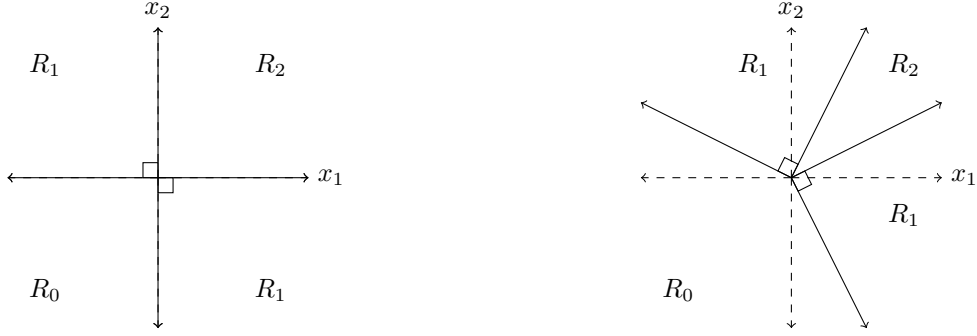

	Within each $R_i$, $\bar{\chi}^2(V,\Lambda_+) \sim \chi^2_i$, where $\chi^2_0$ is a point mass at zero. The weights of the $\bar{\chi}^2_K(V,\Lambda_+)$ are determined by the probability of falling into each partition, and are seen to depend on the covariance $V$. 
	
	Expansive literature exists on the $\ChiBarSq$ distribution.  The paper \citet{K63} introduces the topic in the context of order constrained one-sided tests; Chapter 3 of \citet{SS05} contains detailed examples and derivations as well as a collection of many contemporary results; and \citet{Shap85} discusses the weights  $w_{i}(k, V, C)$ extensively.    
	
	\section{Matching Details for Smoking and Polycyclic Aromatic Hydrocarbons}
	Individuals were classified as cigarette smokers or as non-cigarette-smokers in accordance with the criteria used in \citet{FS16}.  This divided the population of 1638 total individuals into 432 cigarette smokers and 1206 non-cigarette-smokers.  The population of non-smokers did include those who may have smoked in the past but had stopped smoking by the time of the survey, as well as individuals who had never smoked cigarettes.  The individuals were placed into matched groups using a full matching procedure \citep[\S 8.5]{R10}; thus each group contained a single treated unit and multiple control units or a single control unit and multiple treated units.  Pre-treatment covariates were selected based upon recent medical research. To form the fully-matched sets, propensity score caliper with a rank-based Mahalanobis distance for within-caliper distance was used.  The caliper was set at 0.08 and logistic regression was performed to estimate propensity scores \citep[\S 8]{R10}. See \citet[Appendix A]{FS16} for further implementation details.
	
	
	\bibliographystyle{plainnat}
	\bibliography{bibliography}
	
\end{document}